\documentclass{article}

\PassOptionsToPackage{numbers, compress}{natbib}

\usepackage[preprint]{neurips_2026}

\usepackage[utf8]{inputenc} 
\usepackage[T1]{fontenc}    
\usepackage{hyperref}       
\usepackage{url}            
\usepackage{booktabs}       
\usepackage{amsfonts}       
\usepackage{nicefrac}       
\usepackage{microtype}      
\usepackage{xcolor}         

\usepackage{amsmath}
\usepackage{amsthm}
\usepackage[capitalize,noabbrev]{cleveref}
\usepackage{graphicx}
\usepackage{subcaption}
\usepackage{algorithm}
\usepackage{algorithmic}
\usepackage{float}

\DeclareMathOperator*{\argmax}{arg\,max}
\DeclareMathOperator*{\argmin}{arg\,min}

\DeclareMathOperator{\KL}{KL}
\DeclareMathOperator{\TV}{TV}

\usepackage{aliascnt}

\theoremstyle{plain}
\newtheorem{theorem}{Theorem}[section]

\newaliascnt{propcnt}{theorem}

\aliascntresetthe{propcnt}
\crefname{propcnt}{Proposition}{Propositions}

\newaliascnt{lemmacnt}{theorem}
\newtheorem{lemma}[lemmacnt]{Lemma}
\aliascntresetthe{lemmacnt}
\crefname{lemmacnt}{Lemma}{Lemmas}

\newaliascnt{corcnt}{theorem}
\newtheorem{corollary}[corcnt]{Corollary}
\aliascntresetthe{corcnt}
\crefname{corcnt}{Corollary}{Corollaries}

\theoremstyle{definition}
\newaliascnt{defcnt}{theorem}

\aliascntresetthe{defcnt}
\crefname{defcnt}{Definition}{Definitions}

\newaliascnt{assumpcnt}{theorem}
\newtheorem{assumption}[assumpcnt]{Assumption}
\aliascntresetthe{assumpcnt}
\crefname{assumpcnt}{Assumption}{Assumptions}

\theoremstyle{remark}
\newaliascnt{remcnt}{theorem}
\newtheorem{remark}[remcnt]{Remark}
\aliascntresetthe{remcnt}
\crefname{remcnt}{Remark}{Remarks}

\title{Variational predictive resampling}

\author{%
  Laura Battaglia\thanks{Equal contribution.} \\
  Department of Statistics\\
  University of Oxford\\
  \texttt{battaglia@stats.ox.ac.uk} \\
  \And
  Stefano Cortinovis \footnotemark[1] \\
  Department of Statistics\\
  University of Oxford\\
  \texttt{cortinovis@stats.ox.ac.uk} \\
  \And
  Chris Holmes\\
  Department of Statistics\\
  University of Oxford\\
  Ellison Institute of Technology\\
  \texttt{holmes@stats.ox.ac.uk} \\
  \And
  David T. Frazier \\
  Department of Econometrics and Business Statistics\\
  Monash University\\
  \texttt{david.frazier@monash.edu} \\
  \And
  Jack Jewson \\
  Department of Econometrics and Business Statistics\\
  Monash University\\
  \texttt{jack.jewson@monash.edu} \\
}

\begin{document}

\maketitle

\begin{abstract}
    Bayesian inference provides principled uncertainty quantification, but accurate posterior sampling with MCMC can be computationally prohibitive for modern applications.
    Variational inference (VI) offers a scalable alternative and often yields accurate predictive distributions, but cheap variational families such as mean-field (MF) can produce over-concentrated approximations that miss posterior dependence.
    We propose variational predictive resampling (VPR), a scalable posterior sampling method that exploits VI's predictive strength within a predictive-resampling framework to better approximate the Bayesian posterior.
    Given a prior--likelihood pair, VPR repeatedly imputes future observations from the current variational predictive, updates the variational approximation after each imputation, and records the parameter value implied by the completed sample.
    We establish conditions under which the law of the parameter returned by VPR is well defined and show that its finite-horizon approximation converges to this limit.
    In a tractable Gaussian location model, we show that VPR with MF variational predictives converges to the exact Bayesian posterior, whereas the optimal MF-VI approximation retains a non-vanishing asymptotic gap.
    Experiments on linear regression, logistic regression, and hierarchical linear mixed-effects models demonstrate that VPR substantially improves posterior uncertainty quantification and recovers posterior dependence missed by MF-VI, while remaining computationally competitive with, and often more efficient than, MCMC.
\end{abstract}

\section{Introduction}

Bayesian inference plays a central role in modern machine learning \citep{mackay2003information, bishop2006pattern, murphy2023probabilistic} due to its ability to produce principled uncertainty quantification over unknown parameters.
However, performing \emph{exact} Bayesian inference traditionally relies on Monte Carlo sampling methods, most notably MCMC \citep{gelfand1990sampling}.
While a plethora of MCMC algorithms exist, even state-of-the-art methods are often computationally prohibitive for high-dimensional statistics and machine learning applications.

In this context, variational inference (VI) \citep{JordanVI, BealThesis, jaakkola2000bayesian} offers a scalable alternative for \emph{approximate} Bayesian inference.
Variational methods fit a parametric \emph{variational family}, typically by minimising the Kullback--Leibler divergence (KLD) to the Bayesian posterior, with the quality of the approximation depending critically on the family chosen.
Restrictive families such as the \emph{mean-field} (MF) family, which assumes independence across parameters, are cheap to fit but typically produce over-concentrated approximations and miss posterior dependence \citep{Turner_Sahani_2011, blei2017variational, wang2005inadequacy}.
Nonetheless, despite their shortcomings when used directly for posterior inference, cheap variational approximations often induce surprisingly accurate \emph{predictive} distributions \citep{blei2017variational,frazier2025loss}.
In this paper, we exploit VI's predictive strength to construct a scalable sampling method that improves uncertainty quantification relative to the variational posterior itself.

\begin{figure}[H]
    \centering
    \includegraphics[width=0.8\linewidth]{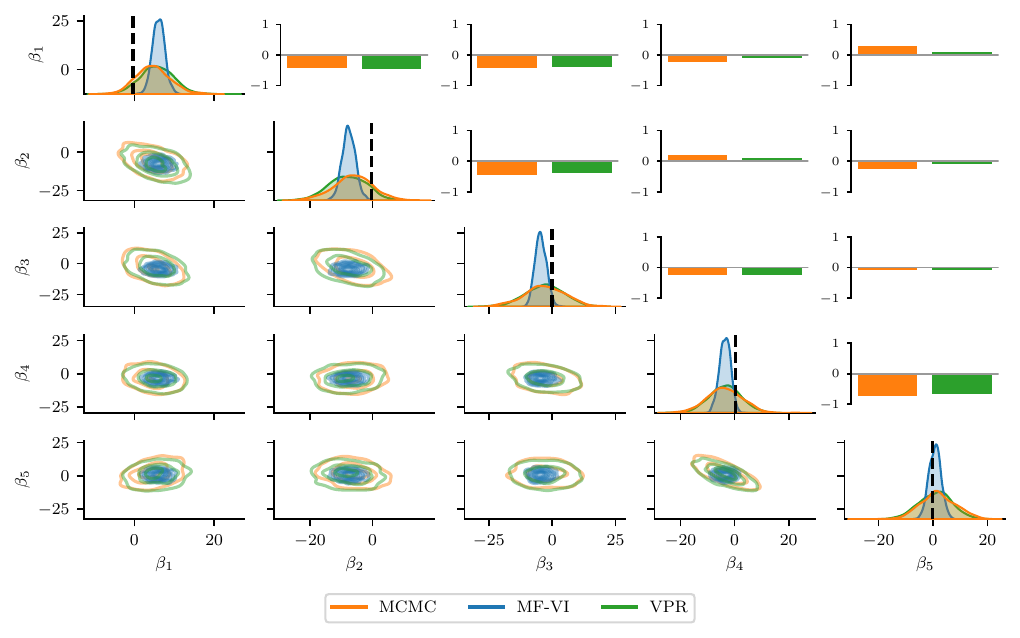}
    \caption{Logistic regression posterior. The lower triangle, diagonal, and upper triangle show, respectively, the KDE bivariate densities, KDE marginal densities, and Pearson's correlations of the $\beta$ posteriors for the MCMC reference, MF-VI and VPR. MF-VI cannot represent posterior dependence, whereas VPR more closely matches MCMC. The black vertical lines denote the true coefficients.}
    \label{fig:logistic_pairplot}
\end{figure}

We do so via the predictive Bayesian framework of \citet{fong2023martingale, fortini2025exchangeability}, which reinterprets posterior uncertainty as reflecting the randomness of an unobserved infinite sequence of future data that, if realised, would pin down the parameters of interest exactly.
Posterior samples are then generated through \emph{predictive resampling}: one repeatedly imputes future data from a one-step-ahead predictive distribution and records the parameter value implied by the completed sample, with the induced parameter distribution known as the \emph{martingale posterior} \citep{fong2023martingale}.
The predictive determines the induced law on the parameter of interest, making it the central modelling choice.

A classical result of \citet{doob1949application} shows that, when the \emph{exact} Bayesian posterior predictive is used, predictive resampling recovers the corresponding Bayesian posterior.
This suggests a natural strategy for scalable posterior approximation: starting from a fixed prior and likelihood, replace the exact Bayesian predictive with its cheap variational counterpart.
We call the resulting method variational predictive resampling (VPR), and show that it substantially improves upon the uncertainty quantification of the variational posterior while often being faster than MCMC.
\Cref{fig:logistic_pairplot} illustrates this in a highly correlated logistic regression example: although the predictive is derived from an MF variational posterior, VPR recovers a dependence structure much closer to the MCMC reference.

The rest of the paper proceeds as follows.
\Cref{sec:background} reviews Bayesian, variational and predictive inference.
\Cref{sec:vpr} introduces VPR, establishes existence of its limiting posterior, and shows that it provably outperforms MF-VI in terms of KLD to the Bayesian posterior in a tractable Gaussian location model.
\Cref{sec:Simulations} reports experiments demonstrating scalability and improved uncertainty quantification, and \cref{sec:discussion} concludes.

\section{Background}
\label{sec:background}
\subsection{Bayesian inference}

Let $y_{1:n} \in \mathcal{Y}^n$ denote observations and let $p(y_i\mid \theta)$ be a likelihood with unknown parameter $\theta \in \Theta \subseteq \mathbb{R}^p$. Given a prior distribution $\pi(\theta)$, Bayesian inferences are derived from the posterior 
\begin{equation}
    \pi(\theta\mid y_{1:n}) \propto \prod_{i=1}^n p(y_i\mid \theta) \pi(\theta).
    \label{eq:bayes-posterior}
\end{equation}
Parameter uncertainty can be propagated to a future $y_{n+1}$ through the posterior predictive distribution
\begin{equation}
    p(y_{n+1}\mid y_{1:n}) = \int_{\Theta} p(y_{n+1}\mid \theta)\pi(\theta\mid y_{1:n}) d\theta.
    \label{eq:bayes-predictive}
\end{equation}
However, outside a limited set of conjugate model--prior pairs, the normalising constant in \eqref{eq:bayes-posterior} is intractable, rendering direct calculation of \eqref{eq:bayes-posterior} and \eqref{eq:bayes-predictive} infeasible.
The standard remedy is Markov chain Monte Carlo (MCMC) \citep{gelfand1990sampling}, which nevertheless scales poorly in both $n$ and $p$ \citep[e.g.,][]{johndrow2020no, bandeira2023free, pozza2025fundamental}.

\subsection{Variational inference}

Variational inference (VI) provides a scalable alternative by approximating the posterior \eqref{eq:bayes-posterior} within a parametric variational family $\mathcal{Q}:= \{q_{\lambda}(\theta): \lambda\in\Lambda\}$.
The variational approximation is typically chosen to minimise the reverse KLD, $\KL[q_\lambda(\theta) \parallel \pi(\theta\mid y_{1:n})]$, which is equivalent to finding the variational parameters $\hat{\lambda}_n$ that maximise the evidence lower bound (ELBO)
\begin{equation*}
    \mathcal L(\lambda; y_{1:n}) = \sum_{i=1}^n \mathbb E_{q_{\lambda}(\theta)}[\log p(y_i\mid \theta)] - \KL[q_\lambda(\theta)\parallel \pi(\theta)].
\end{equation*}
VI performs stochastic optimisation of the ELBO via minibatch estimates of the expected log-likelihood \citep{hoffman2013stochastic}, yielding an explicit approximating density rather than MC samples. The simplest and most widely adopted variational family is the mean-field (MF) family, which factorises over the dimensions of $\theta$: $q_{\lambda}(\theta) = \prod_{j=1}^p q_j(\theta_j;\lambda_j)$.
This restriction enables fast optimisation for $\hat{\lambda}_n$ and often yields accurate point estimates for $\theta$ \citep[e.g.][]{titterington2006convergence, wang_frequentist_2019}.
However, MF combined with reverse KLD minimisation typically underestimates posterior uncertainty and neglects posterior dependence, thereby impairing inference.
Richer variational families \citep[e.g.,][]{hoffman2015stochastic, rezende2015variational, ong2018gaussian} or alternative variational criteria \citep[e.g.,][]{minkadivergence, higgins2017beta, domke2018importance, chen2018variational} can mitigate these issues, but usually at additional computational cost.

Our goal is instead to retain the computational convenience of MF-VI while exploiting the fact that the variational posterior predictive
\begin{equation}
    \tilde{p}_{\hat{\lambda}_n}(y_{n+1}) := \int p(y_{n+1}\mid\theta) q_{\hat{\lambda}_{n}}(\theta) d\theta
    \label{eq:vi-predictive}
\end{equation}
can approximate the Bayesian posterior predictive \eqref{eq:bayes-predictive} well even when $q_{\hat{\lambda}_{n}}(\theta)$ poorly approximates $\pi(\theta\mid y_{1:n})$ \citep{blei2017variational}.
Indeed, \citet{frazier2025loss} recently showed that, under broad conditions, the total variation distance between $p(\cdot\mid y_{1:n})$ and $\tilde{p}_{\hat{\lambda}_n}(\cdot)$ vanishes as $n \rightarrow \infty$. 

\subsection{Predictive Bayes and predictive resampling}\label{sec:mgp}

The uncertainty quantified by \eqref{eq:bayes-posterior} is traditionally interpreted as arising from the Bayesian view that $\theta$ is itself random under the prior $\pi(\theta)$.
The predictive Bayes viewpoint instead attributes posterior uncertainty about $\theta$ to the unobserved sequence $y_{n+1:\infty}$.
This perspective has roots in de Finetti's analysis of exchangeable sequences \citep{de1937prevision} and has been formalised in recent work \citep{fong2023martingale, fortini2025exchangeability}.
Indeed, given infinite realisations from the data generating process, $y_{1:\infty}$, one could follow the maximum likelihood principle and compute 
\begin{equation}
    \hat\theta(y_{1:\infty}):=\argmin_{\theta\in\Theta}\lim_{k\rightarrow\infty}-\frac{1}{k}\sum_{i=1}^{k}\log p(y_i\mid\theta)
    \label{eq:mle}
\end{equation}
exactly, with no residual uncertainty.
That is, uncertainty about $\theta$ given the \textit{sample} $y_{1:n}$ stems from 
not observing the remainder of the \textit{population}, $y_{n+1:\infty}$.

\citet{fong2023martingale} propose to quantify this uncertainty in three steps. First, one specifies a joint predictive distribution $F(\cdot\mid y_{1:n})$ for the missing sequence $y_{n+1:\infty}$, with density
\begin{equation}
    f(\tilde{y}_{1:\infty} \mid y_{1:n}) = \prod_{i=0}^\infty f(\tilde{y}_{i+1} \mid y_{1:n}, \tilde{y}_{1:i}),
    \label{eq:joint_predictive}
\end{equation}
where the tilde distinguishes generated observations from observed ones.
In practice, as shown by \eqref{eq:joint_predictive}, it suffices to specify the one-step-ahead predictive distributions $F(\tilde{y}_{i+1} \mid y_{1:n}, \tilde{y}_{1:i})$, or equivalently their densities $f(\tilde{y}_{i+1} \mid y_{1:n}, \tilde{y}_{1:i})$.
Second, one imputes the missing observations by sampling $\tilde{y}_{1:\infty} \sim F(\cdot\mid y_{1:n})$.
Third, one records the parameter induced by $\tilde{y}_{1:\infty}$, namely 
\begin{equation}
    \theta_{n,\infty}:=\hat\theta(\{y_{1:n}, \tilde{y}_{1:\infty}\})=\argmin_{\theta\in\Theta}\lim_{N\rightarrow\infty}-\tfrac{1}{N+n}\bigl(\textstyle\sum_{i=1}^{n}\log p(y_i\mid\theta)+\sum_{i=1}^{N}\log p(\tilde{y}_i\mid\theta)\bigr).\label{equ:theta_n_infty}
\end{equation}

Here $\theta_{n,\infty}$ is random because it depends on the imputed sequence $\tilde{y}_{1:\infty}\sim F(\cdot\mid y_{1:n})$;
\Cref{sec:validity} gives sufficient conditions on $F$ under which $\theta_{n,\infty}$ is well defined.
Notably, the predictive rule $F$ is the key modelling choice: no prior on $\theta$ needs to be specified. This construction gives rise to the \emph{predictive posterior distribution}, defined as
\begin{equation*}
    \Pi_{\infty}(\theta_{n,\infty} \in A \mid y_{1:n}) =\int \mathbf{1}(\theta_{n,\infty} \in A)\, dF(\tilde{y}_{1:\infty} \mid y_{1:n}),
\end{equation*}
for all measurable sets $A \subseteq \Theta$.
A canonical example is Rubin's Bayesian bootstrap \citep{rubin1981bayesian}, recovered by taking $F$ to be the empirical distribution of the data.
In this case, forward simulation of $\tilde{y}_{1:\infty}$ is equivalent to sampling Dirichlet weights on the observed sample, and draws of $\theta$ only require solving the corresponding weighted M-estimation problem.
However, this $F$ places mass only on the observed $y_{1:n}$, and recent work has therefore focused on more flexible predictive rules.

For more flexible choices of $F$, sampling $\Pi_\infty$ does not collapse to weighted optimisation and must be approximated by sequentially propagating the predictive rule forward from the observed data $y_{1:n}$ to a finite horizon $N$.
This yields $\theta_{n, N} = \hat\theta(y_{1:n}, \tilde{y}_{1:N})$ with law $\Pi_N$.
Such a forward-simulation procedure, known as \emph{predictive resampling}, is presented in \cref{alg:gen_pred_resample}.
\begin{figure}[t]
\begin{minipage}[t]{0.48\textwidth}
\begin{algorithm}[H]
\caption{Predictive resampling \citep{fong2023martingale}}
\begin{algorithmic}[1]
\STATE {\bfseries Input:} Depth $N$, $L$ samples, data $y_{1:n}$, one-step-ahead predictive $F(\cdot\mid\cdot)$.
\FOR{$l = 1$ {\bfseries to} $L$}
    \FOR{$i =0$ {\bfseries to} $N-1$}
        \STATE Sample $\tilde{y}^{(l)}_{i+1} \sim F(\cdot\mid y_{1:n}, \tilde{y}_{1:i}^{(l)})$
        \STATE $F(\cdot\mid y_{1:n}, \tilde{y}_{1:i}^{(l)})\mapsto F(\cdot\mid y_{1:n}, \tilde{y}_{1:i+1}^{(l)})$
    \ENDFOR
    \STATE Form $\theta_{n, N}^{(l)} \gets \hat\theta(\{y_{1:n}, \tilde{y}^{(l)}_{1:N}\})$ 
\ENDFOR
\STATE {\bfseries Return:} $\{\theta_{n, N}^{(l)}\}_{l=1}^L\sim\Pi_N(\cdot\mid y_{1:n})$.
\end{algorithmic}
\label{alg:gen_pred_resample}
\end{algorithm}
\end{minipage}\hfill
\begin{minipage}[t]{0.48\textwidth}
\begin{algorithm}[H]
\caption{Variational predictive resampling}
\begin{algorithmic}[1]
\STATE {\bfseries Input:} Depth $N$, $L$ samples, data $y_{1:n}$.
\STATE Compute $\hat{\lambda}_{n, 0} \gets \argmax \mathcal L(\lambda; y_{1:n})$
\FOR{$l = 1$ {\bfseries to} $L$}
    \STATE Set $\hat{\lambda}^{(l)}_{n, 0} \gets \hat{\lambda}_{n, 0}$
    \FOR{$i =0$ {\bfseries to} $N - 1$}
        \STATE Sample $\tilde{y}^{(l)}_{i+1} \sim \tilde{p}_{\hat{\lambda}^{(l)}_{n,i}}(\cdot)$
        \STATE $\hat{\lambda}^{(l)}_{n, i + 1} \gets \argmax \mathcal L(\lambda; y_{1:n}, \tilde{y}^{(l)}_{1:i+1})$
    \ENDFOR
    \STATE Form $\theta_{n, N}^{(l)} \gets \hat\theta(\{y_{1:n}, \tilde{y}^{(l)}_{1:N}\})$
\ENDFOR
\STATE {\bfseries Return:} $\{\theta_{n, N}^{(l)}\}_{l=1}^L\sim \Pi_N(\cdot\mid y_{1:n})$, approximating samples from $\pi(\theta \mid y_{1:n})$.
\end{algorithmic}
\label{alg:MPVI}
\end{algorithm}
\end{minipage}
\end{figure}
In supervised settings where one is interested in learning about the conditional distribution $Y\mid X$\footnote{In supervised settings, we omit $x$ and $\tilde{x}$ from the conditioning set of likelihood and predictive distributions for ease of notation.} given data $(y_j, x_j)_{j=1}^n$, features $\tilde{x}_{1:N}$ are typically forward-sampled via the bootstrap \citep{efron1992bootstrap} or Bayesian bootstrap \citep{fong2023martingale, fong_bayesian_2024, yung2025moment}.
Crucially, since the $L$ paths are generated independently conditional on $y_{1:n}$, predictive resampling is trivially parallel and yields independent draws from the finite-horizon law $\Pi_N$.
As a result, unlike MCMC, it avoids autocorrelation between samples and can be substantially faster when predictive sampling and updating are cheap.

Existing choices of $F$ include copula and quantile functions \citep{fong2023martingale, fong_bayesian_2024}, Newton's bootstrap \citep{fortini2023prediction}, parametric stochastic gradient descent \citep{holmes_statistical_2023, fong_asymptotics_2024}, KDE updates \citep{battiston2025bayesian}, large language models \citep{falck2024context}, neural networks \citep{wang_uncertainty_2024, wu_posterior_2024}, and tabular foundation models \citep{ng2025tabmgp, nagler_uncertainty_2025}.

\section{Variational predictive resampling}
\label{sec:vpr}

We work in the standard Bayesian setting: given a prior $\pi(\theta)$ and likelihood $p(y \mid \theta)$, our target is the Bayesian posterior $\pi(\theta \mid y_{1:n})$ \eqref{eq:bayes-posterior}.
We propose variational predictive resampling (VPR), an instance of predictive resampling (\cref{alg:gen_pred_resample}) in which the one-step-ahead predictive $F$ is chosen to be the variational predictive
\begin{equation*}
\tilde{p}_{\hat{\lambda}_{n,i}}(\cdot) := \int_{\Theta} p(\cdot\mid \theta) q_{\hat{\lambda}_{n,i}}(\theta)\,d\theta,
\quad\text{where}\quad
\hat{\lambda}_{n,i} := \argmax_{\lambda} \mathcal{L}(\lambda; y_{1:n}, \tilde{y}_{1:i}),\ i \geq 0
\end{equation*}
denotes the variational parameters used to approximate $\pi(\theta\mid y_{1:n}, \tilde{y}_{1:i})$. 
The rationale is \citet{doob1949application}'s martingale argument: predictive resampling driven by the exact Bayesian predictive recovers the Bayesian posterior.
Thus, when the variational predictive $\tilde{p}_{\hat{\lambda}_{n, i}}$ is close to the Bayesian predictive $p(\cdot \mid y_{1:n}, \tilde{y}_{1:i})$, the induced law $\Pi_\infty$ is expected to be close to $\pi(\theta\mid y_{1:n})$. 
In turn, if $N$ is large enough, $\Pi_N$ should resemble $\Pi_\infty$.
\Cref{alg:MPVI} describes the ideal VPR procedure.
The rest of this section establishes conditions under which the limiting law $\Pi_\infty$ exists and $\Pi_N$ converges to it, then studies the approximation to the Bayesian posterior in a tractable Gaussian model, describes the scalable implementation used in our experiments, and finally positions VPR relative to related methods. \looseness=-1

\subsection{Theoretical validity of \cref{alg:MPVI}}\label{sec:validity}

\subsubsection{Existence of $\Pi_\infty$}

We first establish that \cref{alg:MPVI} induces a well-defined limiting distribution $\Pi_{\infty}(\theta\in\cdot\mid y_{1:n})$.
A sufficient condition, due to \citet{battiston2025bayesian}, is that the imputed sequence $\tilde{Y}_{1:\infty}$ is \emph{almost conditionally identically distributed} (a.c.i.d.). 
The sequence of predictive distributions $\{F(\cdot \mid y_{1:n}, \tilde{y}_{1:i})\}_{i \geq 0}$ produces $(\xi_i)$-a.c.i.d.~samples if
\begin{align*}
    \TV&\Bigl(\mathbb{E}_{\tilde{Y}_{i+1}}\left[F(\cdot \mid y_{1:n}, \tilde{y}_{1:i}, \tilde{Y}_{i+1})\mid y_{1:n},\tilde{y}_{1:i} \right],  F(\cdot\mid y_{1:n}, \tilde{y}_{1:i})\Bigr) \leq \xi_i,
\end{align*}
where $\TV(\cdot, \cdot)$ is the total-variation distance.
If $\sum_{i=1}^{\infty} \xi_i < \infty$ almost surely, then $\Pi_{\infty}$ is well defined.
The a.c.i.d.~condition generalises the martingale condition of \citet{fong2023martingale}, where $\xi_i = 0$ for all $i$, under which the sequence is exactly c.i.d.~\citep[see][]{berti2004limit, fortini2023prediction, fortini2025exchangeability}.

\Cref{thm:ACID_VI} shows that the variational predictive sequence generated by \cref{alg:MPVI} is a.c.i.d.
\begin{theorem}
Let $\{\tilde{Y}_{i+1}\}_{i\geq 0}$ be generated from the sequence of variational predictives $\tilde{p}_{\hat{\lambda}_{n, i}}(\cdot)$. Then, provided $\tilde{p}_{\lambda}(\cdot)$ is twice continuously differentiable in $\lambda$, $\{\tilde{Y}_{i+1}\}_{i\geq 0}$ is an $(\xi_i)_{i\geq 0}$-a.c.i.d. sequence with
{
    \small
    \begin{align*}
        \xi_i \leq &\frac{1}{2} \left(\int_{\mathcal{Y}}\left|\nabla_{\lambda}\tilde{p}_{\hat{\lambda}_{n, i}}(y)\right|\,dy\right)^\top\left|\mathbb{E}\left[\Delta \hat{\lambda}_{n,i}\right] \right|
         + \frac{1}{4} \sqrt{\mathbb{E}\left[\left(\textstyle{\int_{\mathcal{Y}}}\left|\kappa_{\lambda^{\dagger}(\tilde{Y}_{i+1})}(y)\right|dy\right)^2\right]}\cdot
         \sqrt{\mathbb{E}\left[\left\|\Delta \hat{\lambda}_{n,i}\right\|_2^4\right]},
    \end{align*}
}
almost surely, where $\Delta \hat{\lambda}_{n,i} :=\hat{\lambda}_{n, i+1}(\tilde{Y}_{i+1}) - \hat{\lambda}_{n, i}$, $\kappa_{\lambda}(y)$ is the eigenvalue of $\nabla^2_{\lambda}\tilde{p}_{\lambda}(y)$ with the largest absolute value,  $\lambda^{\dagger}(\tilde{Y}_{i+1})$ is some intermediate value between $\hat{\lambda}_{n, i}$ and $\hat{\lambda}_{n, i+1}(\tilde{Y}_{i+1})$, and the expectations above are conditional on $y_{1:n}$ and $\tilde{y}_{1:i}$. 
\label{thm:ACID_VI} 
\end{theorem}
The proof in \cref{app:proof_thm_ACID_VI} uses a Taylor expansion argument similar to \citet[Proposition 1]{battiston2025bayesian}.
\Cref{cor:ACID_conditions} gives convenient sufficient conditions for $\Pi_{\infty}$ to be well defined.
\begin{corollary}
    \label{cor:ACID_conditions}
    $\Pi_{\infty}$ is well defined if $\left(\int_{\mathcal{Y}}\left|\nabla_{\lambda}\tilde{p}_{\hat{\lambda}_{n, i}}(y)\right|dy\right)^\top |\mathbb{E}[\Delta \hat{\lambda}_{n,i}]|$ 
    and $\sqrt{\mathbb{E}[(\int_{\mathcal{Y}}\left|\kappa_{\lambda^{\dagger}(\tilde{Y}_{i+1})}(y)\right|dy)^2]}\cdot \sqrt{\mathbb{E}[\|\Delta \hat{\lambda}_{n,i}\|_2^4]}$ are both almost surely $O(1/i^{1+\epsilon})$ for some $\epsilon > 0$.
\end{corollary}
\Cref{cor:ACID_conditions} requires that the variational parameters move little when updated after one additional generated observation. 
In short, $\Pi_{\infty}$ is well defined whenever these updates are sufficiently stable along the imputed path.
\Cref{sec:verification_acid} verifies these conditions in a simple example.  

\subsubsection{Posterior behaviour as $N\rightarrow\infty$}

\Cref{thm:ACID_VI,cor:ACID_conditions} provide conditions under which $\Pi_\infty$ exists.
In practice, however, we can approximate $\Pi_\infty$ by the finite-horizon law $\Pi_N$.
\Cref{lem:one} shows that, under standard conditions, the random variable $\theta_{n,N}$ underlying $\Pi_N$ converges to the random variable $\theta_{n,\infty}$ underlying $\Pi_\infty$.
\begin{lemma}\label{lem:one}
Under \cref{ass:pred_dist,ass:regularity}, there exists a unique $\theta_{n,\infty}\in\Theta$ such that $\|\theta_{n,N}-\theta_{n,\infty}\|=O_p(1/\sqrt{N})$ as $N\rightarrow\infty$, where $\theta_{n,N}\sim\Pi_N$ and $\theta_{n,\infty}\sim\Pi_\infty$.
\end{lemma}
A consequence of \Cref{lem:one} is that $\Pi_N$ converges weakly to $\Pi_{\infty}$, validating its use as a computationally feasible approximation.
\Cref{lem:one} generalises existing parametric results, most notably \citet[Corollary 1]{fong_asymptotics_2024} (based on the parametric parameter update in \citet{holmes_statistical_2023}). Unlike \citet{fong_asymptotics_2024}, however, we do not require that the predictive sequence is a martingale, as such a condition is violated in many cases. Instead, we rely on \cref{thm:ACID_VI} and novel asymptotic arguments (proof in \cref{app:proof_theta_conv}). 

An additional useful feature of \cref{lem:one} is that it implicitly encodes how $N$ must scale with $n$.
In particular, the proof of \cref{lem:one} shows that, for the error $\|\theta_{n,N}-\theta_{n,\infty}\|$ to be controllable, we require $n/N=o(1)$.
Otherwise, as $n$ becomes large, there is no reason to believe that $\theta_{n,N}$ targets $\theta_{n,\infty}$: the observed $y_{1},\dots, y_n$ have equal weight in computing $\theta_{n,N}$ relative to the simulated $\tilde{y}_1,\dots,\tilde{y}_N$.
In practice, this requirement can be enforced by choosing $N=\lceil n^{1+\delta} \rceil$ for some $\delta>0$.

\subsection{Posterior approximation}\label{sec:post-approx}

Having established conditions for the existence of $\Pi_{\infty}$ for any variational family satisfying \cref{cor:ACID_conditions}, we specialise to MF variational predictives; henceforth VPR refers to this instantiation. We now consider whether $\Pi_{\infty}$ is close to the Bayesian posterior \eqref{eq:bayes-posterior}.
Formalising this in full generality is challenging; we thus study an analytically tractable multivariate Gaussian location model.

For $y_1,\dots,y_n \in \mathbb{R}^p$, consider $p(y_i\mid \theta) = \mathcal{N}(y_i; \theta, A)$, with unknown mean $\theta\in\mathbb{R}^p$ and known non-diagonal covariance $A$, and let $B = A^{-1}$. Given prior $\theta \sim \mathcal{N}(0, I)$, the posterior is $\pi(\theta \mid  y_{1:n}) = \mathcal{N}(\mu_{n,0}, \Sigma_{n,0})$ and the Gaussian MF approximation is $q_{\hat\lambda_{n,0}}(\theta) = \mathcal{N}(m_{n,0}, \mathrm{diag}(s_{n,0}))$, with
\begin{equation*}
    \Sigma_{n, 0} = (I + nB)^{-1}, \quad \mu_{n, 0} = \Sigma_{n, 0} B \textstyle\sum_{j=1}^n y_j, \quad m_{n, 0} = \mu_{n, 0}, \quad s_{n, 0} = \mathrm{diag}(\Sigma_{n, 0}^{-1})^{-1}.
\end{equation*}
MF-VI therefore matches the posterior mean but fails to capture the posterior dependence induced by the non-diagonal covariance $A$.
Because of this, the KLD between the optimal MF-VI approximation and the posterior does not vanish asymptotically, as \cref{lem:KLD_MFVI} shows.
\begin{lemma}
\begin{equation*}
    \lim_{n \to \infty} \KL\left[q_{\hat{\lambda}_{n, 0}}(\cdot) \parallel \pi(\cdot\mid y_{1:n})\right] = -\frac{1}{2} \log \left(\frac{|B|}{\prod_i B_{ii}}\right) > 0.
\end{equation*}
\label{lem:KLD_MFVI}
\end{lemma}

By contrast, \cref{thm:KLD_MPVI} shows that the KLD between the limiting VPR density, denoted by $\pi_{\infty}$, and the posterior vanishes for this model.
\begin{theorem}
\begin{equation*}
    \KL\left[\pi_\infty(\cdot\mid y_{1:n}) \parallel \pi(\cdot\mid y_{1:n})\right] = O(n^{-2}) \text{ as } n\to\infty.
\end{equation*}
\label{thm:KLD_MPVI}
\end{theorem}
Proofs of both results are in \cref{app:proof_lem_KLD_MFVI,app:proof_MPVI_gaussian_mean}. \Cref{fig:KLD_VI_vs_MP} verifies \cref{lem:KLD_MFVI,thm:KLD_MPVI}: for the chosen $A$, the MF-VI KLD quickly approaches its non-zero limit, while the VPR KLD is already close to zero for small $n$.

\subsection{Scalability}\label{sec:scalability}

Like \cref{alg:gen_pred_resample}, \cref{alg:MPVI} is trivially parallelisable.
Once $q_{\hat{\lambda}_{n,i}}$ is obtained, sampling from $\tilde p_{\hat{\lambda}_{n,i}}(\cdot)$ is also straightforward: draw \smash{$\theta\sim q_{\hat{\lambda}_{n,i}}(\cdot)$} and then $\tilde y_{i+1}\mid\theta\sim p(\cdot\mid \theta)$.
The two main computational bottlenecks are therefore the repeated re-estimation of the variational parameters \smash{$\hat{\lambda}_{n, i+1}$} and the final computation of \smash{$\theta^{(l)}_{n,N}$}.

When the ELBO maximiser is unavailable in closed form, the ideal update in \cref{alg:MPVI} requires solving a new variational optimisation problem after every generated observation, which would be prohibitive if done to convergence at each of the $N$ resampling steps.
Instead, we approximate each update by taking a small, fixed number $S$ of minibatch stochastic-gradient steps, warm-started at the previous variational parameter \smash{$\hat{\lambda}_{n,i}$}.
This yields a fast local update from \smash{$\hat{\lambda}_{n,i}$} to \smash{$\hat{\lambda}_{n,i+1}$} and provides a scalable approximation to the ideal finite-horizon procedure.
Other approximate update rules could also be used.
We experimented with natural-gradient batch updates and BONG \citep{jones2024bayesian}, a popular single-observation online-VI method.
In preliminary experiments, natural-gradient updates gave similar posterior approximations without improving speed, while BONG yielded poorer approximations to $\pi(\theta\mid y_{1:n})$.
This may be because online-VI methods are designed for streams of genuine observations from the data-generating process, whereas the additional observations in VPR are simulated from the current predictive and contain no new information beyond $y_{1:n}$.
Consequently, approximation errors can compound along the resampling path without being offset by new signal.

Even with a cheap \smash{$\hat{\lambda}_{n,i+1}$} update, computing \smash{$\theta^{(l)}_{n,N}$} in \cref{alg:MPVI} requires solving a final maximum-likelihood problem on $\{y_{1:n}, \tilde{y}_{1:N}\}$ for each resampling path, which may be expensive.
To avoid this, we approximate \smash{$\theta^{(l)}_{n,N}$} by the terminal variational mean \smash{$m^{(l)}_{n,N}$}.
Under standard regularity conditions, this is a consistent and asymptotically normal estimator of $\theta^{(l)}_{n,N}$ \citep{ray2022variational, zhang2020convergence};
moreover, as $N$ grows, the KLD component of the ELBO becomes negligible relative to the likelihood term, so maximising the ELBO approaches maximising the log-likelihood. \Cref{alg:scale_MPVI} in \cref{app:scalable-vpr} describes the resulting scalable implementation of \cref{alg:MPVI}.

\subsection{Positioning in the literature}

We now position VPR relative to the adjacent literature. 
Although variational approximations appear throughout \cref{alg:MPVI,alg:scale_MPVI}, VPR is \emph{not} itself a new variational objective or variational family.
Rather, it is motivated by the observation that cheap variational families are often more reliable for prediction than inference.
VPR exploits this asymmetry within a predictive Bayesian framework: it uses MF-VI as a scalable predictive plug-in, but returns samples that better approximate the Bayesian posterior than the underlying variational posterior itself.

VPR sits within the predictive-resampling literature, with the closest relatives being the parametric methods of \citet{holmes_statistical_2023, fong_asymptotics_2024}.
However, contrary to existing methods, VPR targets \emph{the} Bayesian posterior \eqref{eq:bayes-posterior} induced by a specified prior--likelihood pair.
By contrast, predictive resampling was introduced as a way to elicit \emph{a posterior} directly from a sequence of prediction rules, without the need to specify a prior \citep{fong2023martingale};
while for certain classes of prediction rules the resulting distribution can be interpreted only asymptotically as \emph{a Bayesian posterior}, this is under an unknown implicit prior that is not specified by the user.
To our knowledge, VPR is the first method to use predictive resampling to obtain a scalable approximation to the Bayesian posterior under a given prior--likelihood pair, rather than a more general notion of uncertainty.

It is useful to distinguish predictive resampling from a
parallel stream of work that extends the Bayesian bootstrap in a different direction,  the \emph{posterior bootstrap}: rather than imputing future observations, one reweights observed, and optionally synthetic, data via Dirichlet weights and solves a weighted M-estimation problem for each posterior draw \citep{newton_approximate_1994, lyddon_general_2019, fong2019scalable, dellaporta2022robust}.
The resulting law does not in general coincide with the Bayesian posterior arising from a user-specified parametric prior--likelihood pair.
In particular, one application in \citet{lyddon2018nonparametric} claims a VI correction, but the correction operates by reducing VI's influence rather than improving on it: they target a misspecification-robust posterior on $\theta$, with a concentration parameter $c$ interpolating between a Bayesian bootstrap ($c=0$) and a parametric Bayesian posterior ($c \to \infty$, approximated by MF-VI when too costly); their best results occur as $c \to 0$, where VI is suppressed entirely (their Figure~1).

\section{Results}
\label{sec:Simulations}

This section investigates the performance of finite-$N$ VPR on three inference scenarios of increasing difficulty: (i) conjugate Bayesian linear regression, (ii) non-conjugate single-level logistic regression, and (iii) a hierarchical linear random effects model.
As all are regression problems, the application of VPR requires propagating the features $x$ forward as in \cref{sec:mgp}, which we do by using the bootstrap for all experiments.
Across all settings, we compare VPR to a mean-field variational baseline (MF-VI) sharing the same structure as the one used for VPR, and to a state-of-the-art MCMC reference  \citep[i.e., NUTS,][]{hoffman2014no}\footnote{For all methods, we use highly efficient JAX code that takes advantage of TPU hardware. In particular, for NUTS we employ NumPyro's state-of-the-art implementation \citep{phan2019composable}.}, which we treat as a proxy for the target posterior when the latter is unavailable.
In the conjugate setting, we include the true posterior as an additional baseline, and compute the optimal MF-VI parameters analytically, using them for both the MF-VI baseline and VPR updates.

We evaluate performance in terms of posterior quality and computational efficiency.
In simulated regimes, posterior quality is assessed via empirical marginal credible-interval coverage at nominal level $\alpha = 0.1$. On real data, where the truth is unavailable, we instead compare posterior samples to the MCMC reference via squared maximum mean discrepancy ($\text{MMD}^2$), and report negative log predictive density (NLPD) on a held-out test set as a predictive check.
Efficiency is measured as number of independent samples per second: effective sample size (ESS) for MCMC and number of paths $L$ for VPR, since independently simulated terminal paths at large $N$ can be viewed as independent draws from the limiting law approximating $\Pi_{\infty}$, so we use the number of paths $L$ as the analogous effective sample count.
Methods setup and datasets are described in \cref{app:experimental_details}.

\subsection{Conjugate model: linear regression}
\label{sec:lr_mpopt}

Because Bayesian linear regression with a Gaussian prior on the coefficients $\beta \in \mathbb{R}^d$ and known observation variance $\sigma^2$ is conjugate, it is an ideal first test case for VPR: it provides (i) an exact multivariate Gaussian posterior as ground truth and (ii) closed-form optimal MF-VI updates that can be used to run exact VPR (\cref{alg:MPVI}) without optimisation error.

In this setting, we design a data generating process that makes the inference problem non-trivial and increasingly computationally demanding for VPR and MCMC.
In particular, for increasing $d$, we generate $n = 3 \times d$ data points ensuring a badly conditioned posterior (see \cref{app:linear_data}).
For each $d$, we fix $\beta^*$ and $\sigma^2$, and report the average coverage and efficiency metrics over 100 datasets in \cref{fig:lr_cov_eff}.
MF-VI exhibits pronounced under-coverage that worsens as $d$ increases, while VPR with closed-form updates tracks the true posterior coverage closely and achieves substantially higher efficiency than MCMC.
Although obtained in an idealised setting, these results suggest that VPR can offer a better efficiency/calibration trade-off than state-of-the-art MCMC in challenging high-dimensional regimes, provided that a sufficiently efficient predictive update rule is employed.

\begin{figure}[H]
    \begin{subfigure}[t]{0.48\linewidth}
        \centering
        \includegraphics[trim= {0.0cm 0.00cm 0.0cm 0.0cm}, clip, width=\linewidth, height=4cm, keepaspectratio]{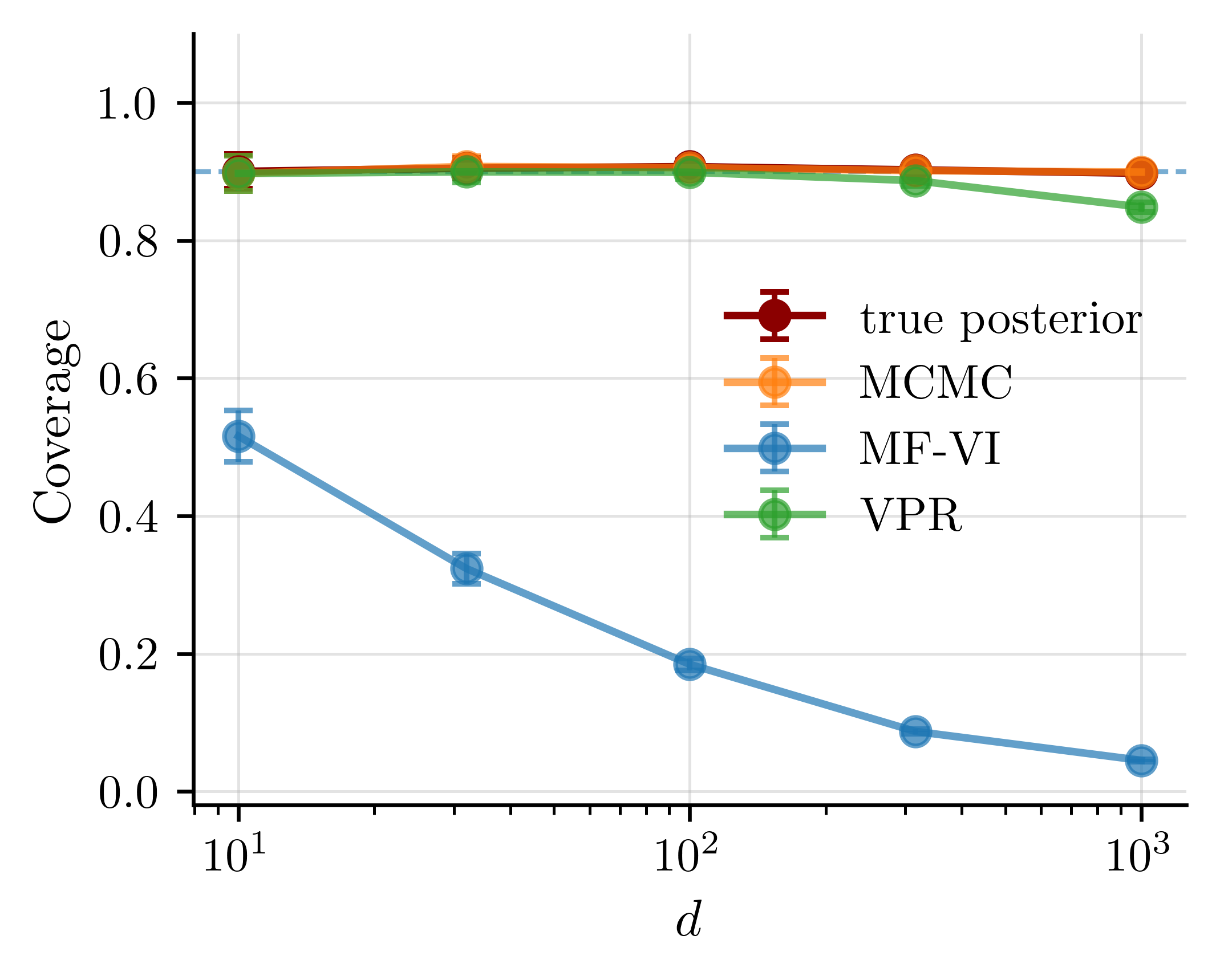}
        \label{fig:lr_coverage} 
    \end{subfigure}%
    \quad
    \begin{subfigure}[t]{0.48\linewidth}
    \centering
    \includegraphics[trim= {0.0cm 0.00cm 0.0cm 0.0cm}, clip, width=\linewidth, height=4cm, keepaspectratio]{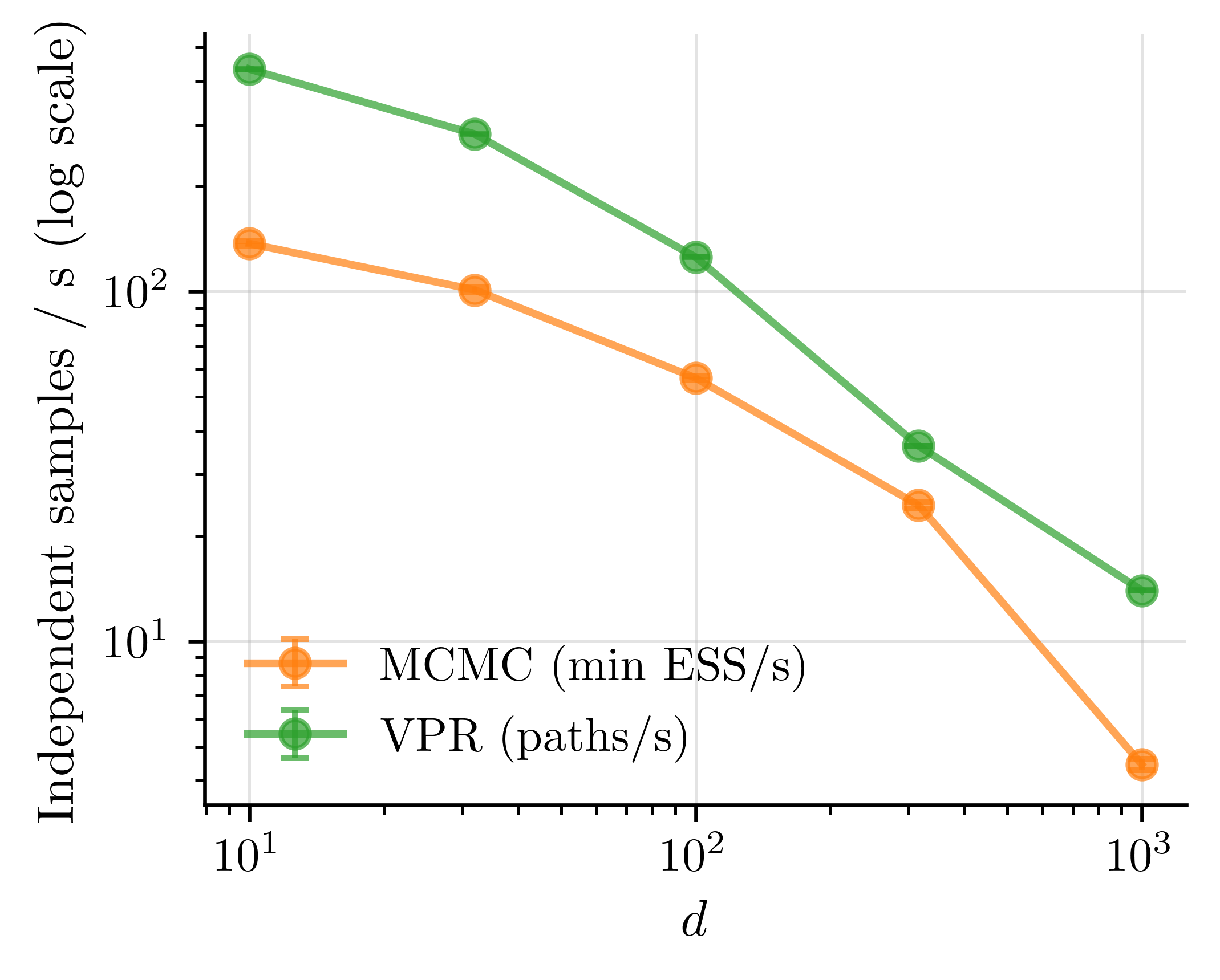}
    \label{fig:lr_efficiency}
    \end{subfigure}
    \caption{Left: Empirical coverage vs. $d$. Right: Sampling efficiency vs.~$d$ (log-scale y-axis): min ESS/s over coefficients for MCMC, paths/s for VPR.
$95\%$ CI error bars, $n = 3d$.}
\label{fig:lr_cov_eff}
\end{figure}

\subsection{Non-conjugate model: logistic regression}\label{sec:logreg_datasets}
 
\begin{table*}[tbp]
\centering
\caption{Logistic regression on real data. Mean $\pm$ 95\% CI over 100 train/test splits; $n=100$.}
\label{tab:logistic_real_data}
\begin{tabular}{lcccccc}
\toprule
\multicolumn{1}{l}{} & \multicolumn{2}{c}{NLPD ratio} & \multicolumn{2}{c}{MMD$^2$ vs MCMC} & \multicolumn{2}{c}{Speed} \\
\cmidrule(lr){2-3}
\cmidrule(lr){4-5}
\cmidrule(lr){6-7}
Dataset ($d$) & MF/MCMC & VPR/MCMC & MF-VI & VPR & ESS/s & paths/s \\
\midrule
german (20) & 1.008{\scriptsize$\pm$0.001} & 1.008{\scriptsize$\pm$0.002} & 0.012{\scriptsize$\pm$0.000} & \textbf{0.006{\scriptsize$\pm$0.000}} & 202{\scriptsize$\pm$6} & \textbf{226{\scriptsize
$\pm$1}} \\
mozilla4 (5) & 1.003{\scriptsize$\pm$0.001} & \textbf{1.001{\scriptsize$\pm$0.001}} & 0.019{\scriptsize$\pm$0.001} & \textbf{0.003{\scriptsize$\pm$0.000}} & 186{\scriptsize$\pm$6} & \textbf{247{\scriptsize$\pm$2}} \\
phoneme (5) & \textbf{1.001{\scriptsize$\pm$0.000}} & 1.002{\scriptsize$\pm$0.001} & 0.007{\scriptsize$\pm$0.001} & \textbf{0.004{\scriptsize$\pm$0.001}} & \textbf{249{\scriptsize$\pm$8}} & 246{\scriptsize$\pm$1} \\
skin (3) & 1.033{\scriptsize$\pm$0.003} & \textbf{1.006{\scriptsize$\pm$0.001}} & 0.111{\scriptsize$\pm$0.006} & \textbf{0.006{\scriptsize$\pm$0.001}} & 144{\scriptsize$\pm$3} & \textbf{250{\scriptsize$\pm$2}} \\
telescope (10) & \textbf{0.998{\scriptsize$\pm$0.001}} & 1.002{\scriptsize$\pm$0.001} & 0.045{\scriptsize$\pm$0.001} & \textbf{0.005{\scriptsize$\pm$0.000}} & 173{\scriptsize$\pm$5} & \textbf{244
{\scriptsize$\pm$2}} \\
\bottomrule
\end{tabular}
\end{table*}

Bayesian logistic regression with a multivariate Gaussian prior on the coefficients $\beta \in \mathbb{R}^d$ is a classical non-conjugate inference example where MF-VI is not available in closed form.
As a result, the ideal \cref{alg:MPVI} is approximated via minibatch gradient descent as in \cref{alg:scale_MPVI}. As a qualitative illustration, we first consider the simulated highly correlated example shown in \cref{fig:logistic_pairplot}. Despite being built from MF variational predictive updates, VPR produces posterior pair plots that closely match the MCMC reference, whereas MF-VI misses the posterior dependence structure. \Cref{app:logistic_experiments} further visualises the evolution of one variational mean across resampling paths for this example; after initial variability, the trajectories stabilise after $\sim$1000 predictive-resampling steps.

We then turn to real-data experiments. \Cref{tab:logistic_real_data} reports results on a range of standard datasets. MF-VI performs comparably to VPR and MCMC in terms of NLPD: this matches the motivation for VPR, which is aimed at settings where MF-VI predicts well, but fails to capture the dependence structure and uncertainty of the posterior. On the other hand, VPR improves substantially over MF-VI in terms of posterior quality, while remaining computationally competitive with MCMC. \Cref{app:logistic_experiments} shows that these conclusions are robust to the choice of the number of gradient steps per predictive update, and that using fewer such steps can yield additional speedups, making VPR much faster than MCMC across all datasets considered. 
\Cref{app:logistic_experiments} also focuses on one real dataset and conducts a thorough ablation study of $N,S, L$ and batch size, demonstrating general robustness of results to hyperparameters.

\subsection{Hierarchical model: linear mixed-effects}
\label{sec:random-effects}

The linear mixed random-effects (LMRE) model represents a popular hierarchical formulation for longitudinal and clustered data: a shared vector of fixed effects $\beta$ captures population-level effects, while group-specific random effects $u_g$ capture latent heterogeneity across clusters \citep{gelman2007data}. 
These latents induce posterior dependence and make exact posterior inference analytically intractable.
Concretely, for groups $g=1,\dots,G$, we consider
\[
\begin{aligned}
y_g \;=\; X_g\beta + Z_g u_g + \varepsilon_g,\quad u_g \;\sim\; \mathcal{N}(0,D), \qquad 
\varepsilon_g \sim \mathcal{N}(0,\sigma^2 I_{n_g}),
\end{aligned}
\]
where $y_g\in\mathbb{R}^{n_g}$, $X_g\in\mathbb{R}^{n_g\times d}$ and $Z_g\in\mathbb{R}^{n_g\times r}$ are known designs, $\beta\in\mathbb{R}^d$ are fixed-effect coefficients, and $u_g\in\mathbb{R}^r$ are latent random effects with covariance $D$.
We set $d=r=3$, $G=4$, fix $\sigma^2$, and use varying group sizes for a total of $n = 82$ observations (see \cref{app:lmre_data} for details).
We use a multivariate Gaussian prior for $\beta$ and an inverse-Wishart prior for $D$. 
As before, ideal VPR is approximated via \cref{alg:scale_MPVI}. 
\Cref{tab:lmre_sim_coverage_eff} summarises average coverage and efficiency metrics over $100$ repetitions for both the population-level $\beta$ and group-level $u$. In this hierarchical regime, VPR delivers uncertainty estimates that are close to the reference posterior while achieving substantially higher effective-sample efficiency.
Additional posterior diagnostics and pair plots are in \cref{app:lmre_experiments}. \looseness=-1

\begin{table}[H]
    \centering
    \caption{LMRE: $90\%$ marginal coverage, averaged over $d$ for $\beta$, over $g, d$ for $u$ and speed; $95\%$ CI.}
    \label{tab:lmre_sim_coverage_eff}
    \begin{tabular}{lccc}
\toprule
Method & $\beta$ & $u$ & Speed \\
\midrule
MF-VI & 0.633{\scriptsize$\pm$0.067} & 0.651{\scriptsize$\pm$0.034} & -- \\
VPR & \textbf{0.900{\scriptsize$\pm$0.042}} & 0.856{\scriptsize$\pm$0.025} & \textbf{134{\scriptsize$\pm$0\,paths/s}} \\
MCMC & 0.927{\scriptsize$\pm$0.038} & \textbf{0.892{\scriptsize$\pm$0.021}} & 28{\scriptsize$\pm$1\,ESS/s} \\
\bottomrule
\end{tabular}
\end{table}

\section{Discussion}
\label{sec:discussion}

VPR uses a standard MF-VI fit on the user's prior--likelihood pair as a predictive plug-in within predictive resampling.
The resulting samples approximate the exact Bayesian posterior \eqref{eq:bayes-posterior} and substantially improve on MF-VI's uncertainty quantification, often at a fraction of the cost of MCMC.
Speedups depend on the cost of the variational updates and on available parallelism, and benchmarks are made against a strong JAX/NUTS implementation under our compute budget.
Ablations  (\cref{app:logistic_experiments}) show that VPR's posterior quality is stable across reasonable ranges of the remaining hyperparameters once $N$ is moderately large, so practitioners can adopt VPR without delicate tuning. 
A central theoretical contribution is establishing that the VPR limiting distribution $\Pi_\infty$ is well-defined under explicit conditions on the variational family and update rule, mirroring the foundational existence results that underpin the predictive Bayes literature \citep{fong2023martingale, battiston2025bayesian, fortini2025exchangeability}.
We then characterise the rate at which the finite-horizon law $\Pi_N$ approaches $\Pi_\infty$.
A tractable Gaussian location model further allows us to show why VPR works: the KLD between $\Pi_\infty$ and the Bayesian posterior vanishes at rate $O(n^{-2})$, even though the MF-VI KLD does not vanish.
Extending the latter result to non-conjugate models -- quantifying how closely the variational predictive must track the Bayesian predictive for $\Pi_\infty$ to remain close to $\pi(\theta\mid y_{1:n})$ -- is a natural direction for future work. 
Empirically, our experiments target the regime motivating VPR: settings where MF-VI is predictively useful but overconfident for inference. 
Applying VPR to more complex models and multi-modal posteriors is a promising next step, which entails carefully analysing the trade-off between richer variational families and additional computational cost.

\begin{ack}
Special thanks to Shahine Bouabid for his invaluable assistance with both the coding aspects of this paper and his recommendations on clarity. We also thank Lorenzo Cappello for useful discussion. 
Laura Battaglia is supported by the Clarendon Funds Scholarship. Stefano Cortinovis is supported
by the EPSRC Centre for Doctoral Training in Modern Statistics and Statistical Machine Learning
(EP/S023151/1). 
\end{ack}

\bibliographystyle{plainnat}
\bibliography{references}

\newpage
\appendix
\setcounter{figure}{0}
\renewcommand{\thefigure}{S\arabic{figure}}
\setcounter{table}{0}
\renewcommand{\thetable}{S\arabic{table}}
\setcounter{algorithm}{0}
\renewcommand{\thealgorithm}{S\arabic{algorithm}}

\crefalias{section}{appendix}
\crefalias{subsection}{appendix}
\crefalias{subsubsection}{appendix}

\section{Proofs of theoretical contributions}\label{app:proofs}

This section contains proofs and regularity conditions for all of the results presented in \cref{sec:vpr}.

\subsection{Proof of \cref{thm:ACID_VI}} \label{app:proof_thm_ACID_VI}

\begin{proof}
    Firstly, write the total-variation distance (TVD) in terms of densities
\begin{align}
    &\TV\left(\mathbb{E}\left[\tilde{P}_{\hat{\lambda}_{n, i+1}}(\tilde{Y}_{i+2} \in \cdot\mid \tilde{Y}_{i+1})\right], \tilde{P}_{\hat{\lambda}_{n, i}}(\tilde{Y}_{i+1} \in \cdot)\right) \nonumber\\
    &= \frac{1}{2}\int_{\mathcal{Y}_2}\left|\mathbb{E}\left[\tilde{p}_{\hat{\lambda}_{n, i+1}}(\tilde{Y}_{i+2} = y\mid \tilde{Y}_{i+1})\right] - \tilde{p}_{\hat{\lambda}_{n, i}}(\tilde{Y}_{i+1} = y)\right|dy\label{Equ:TVD}
\end{align}
where
\begin{align*}
    &\tilde{p}_{\hat{\lambda}_{n, i}}(\tilde{Y}_{i+1} = y) = \int_{\Theta_1} p(\tilde{Y}_{i+1} = y;\theta_1)q_{\hat{\lambda}_{n, i}}(\theta_1)d\theta_1\\
    \mathbb{E}&\left[\tilde{p}_{\hat{\lambda}_{n, i+1}}(\tilde{Y}_{i+2} = y\mid \tilde{Y}_{i+1})\right] =\\
    &\int_{\mathcal{Y}_1}\int_{\Theta_2}p(\tilde{Y}_{i+2} = y;\theta_2)q_{\hat{\lambda}_{n, i+1}(\tilde{y}_{i+1})}(\theta_2)d\theta_2\int_{\Theta_1}p(\tilde{Y}_{i+1} = \tilde{y}_{i+1};\theta_1)q_{\hat{\lambda}_{n, i}}(\theta_1)d\theta_1d\tilde{y}_{i+1},
\end{align*}
and $\theta_1$ and $\theta_2$ are two replications of the same random variable $\theta$. We use this notation to keep track of integrating over the variational posterior twice.

Note that even after integrating $y$ and taking expectations over $\tilde{Y}_{i+1}$, \eqref{Equ:TVD} is a random variable, with dependence on random $\tilde{y}_{1:i}$ (through $\hat{\lambda}_{n, i}$) suppressed for notational convenience throughout. As a result, all equality and inequality statements that follow hold almost surely.

Now if we second order Taylor expand $q_{\hat{\lambda}_{n, i+1}(\tilde{y}_{i+1})}(\theta_2)$ about $\hat{\lambda}_{n, i}$ we get 
\begin{align*}
    q_{\hat{\lambda}_{n, i+1}(\tilde{y}_{i+1})}(\theta_2) &= q_{\hat{\lambda}_{n, i}}(\theta_2) + \nabla_{\lambda}q_{\hat{\lambda}_{n, i}}(\theta_2)^\top\left(\hat{\lambda}_{n, i+1}(\tilde{y}_{i+1}) - \hat{\lambda}_{n, i}\right)\\
    &\qquad + \frac{1}{2}\left(\hat{\lambda}_{n, i+1}(\tilde{y}_{i+1}) - \hat{\lambda}_{n, i}\right)^\top\nabla^2_{\lambda}q_{\lambda^{\dagger}(\tilde{y}_{i+1})}(\theta_2)\left(\hat{\lambda}_{n, i+1}(\tilde{y}_{i+1}) - \hat{\lambda}_{n, i}\right)
\end{align*}
where $\lambda^{\dagger}(\tilde{y}_{i+1})$ is a point between $\hat{\lambda}_{n, i}$ and $\hat{\lambda}_{n, i+1}(\tilde{y}_{i+1})$.

As a result 
\begin{align*}
    &\mathbb{E}\left[\tilde{p}_{\hat{\lambda}_{n, i+1}}(\tilde{Y}_{i+2} = y\mid \tilde{Y}_{i+1})\right]\\
    =& \int_{\mathcal{Y}_1}\int_{\Theta_2}p(\tilde{Y}_{i+2} = y;\theta_2)q_{\hat{\lambda}_{n, i+1}(\tilde{y}_{i+1})}(\theta_2)d\theta_2\int_{\Theta_1}p(\tilde{Y}_{i+1} = \tilde{y}_{i+1};\theta_1)q_{\hat{\lambda}_{n, i}}(\theta_1)d\theta_1d\tilde{y}_{i+1}\\
    =& \int_{\mathcal{Y}_1}\int_{\Theta_2}p(\tilde{Y}_{i+2} = y;\theta_2)q_{\hat{\lambda}_{n, i}}(\theta_2)d\theta_2\int_{\Theta_1}p(\tilde{Y}_{i+1} = \tilde{y}_{i+1};\theta_1)q_{\hat{\lambda}_{n, i}}(\theta_1)d\theta_1d\tilde{y}_{i+1}\\
    &+\int_{\mathcal{Y}_1}\int_{\Theta_2}p(\tilde{Y}_{i+2} = y;\theta_2)\nabla_{\lambda}q_{\hat{\lambda}_{n, i}}(\theta_2)^\top\left(\hat{\lambda}_{n, i+1}(\tilde{y}_{i+1}) - \hat{\lambda}_{n, i}\right)d\theta_2\int_{\Theta_1}p(\tilde{Y}_{i+1} = \tilde{y}_{i+1};\theta_1)q_{\hat{\lambda}_{n, i}}(\theta_1)d\theta_1d\tilde{y}_{i+1}\\
    &+\frac{1}{2}\int_{\mathcal{Y}_1}\int_{\Theta_2}p(\tilde{Y}_{i+2} = y;\theta_2)\left(\hat{\lambda}_{n, i+1}(\tilde{y}_{i+1}) - \hat{\lambda}_{n, i}\right)^\top\nabla_{\lambda}^2q_{\lambda^{\dagger}(\tilde{y}_{i+1})}(\theta_2)\left(\hat{\lambda}_{n, i+1}(\tilde{y}_{i+1}) - \hat{\lambda}_{n, i}\right)d\theta_2\\
    &\quad\cdot\int_{\Theta_1}p(\tilde{Y}_{i+1} = \tilde{y}_{i+1};\theta_1)q_{\hat{\lambda}_{n, i}}(\theta_1)d\theta_1d\tilde{y}_{i+1}\\
    =& \tilde{p}_{\hat{\lambda}_{n, i}}(\tilde{y}_{i+1} = y)\\
    &+\left(\int_{\Theta_2}p(\tilde{Y}_{i+2} = y;\theta_2)\nabla_{\lambda}q_{\hat{\lambda}_{n, i}}(\theta_2)d\theta_2\right)^\top\mathbb{E}\left[\hat{\lambda}_{n, i+1}(\tilde{y}_{i+1}) - \hat{\lambda}_{n, i}\right]\\
    &+\frac{1}{2}\int_{\mathcal{Y}_1}\int_{\Theta_2}p(\tilde{Y}_{i+2} = y;\theta_2)\left(\hat{\lambda}_{n, i+1}(\tilde{y}_{i+1}) - \hat{\lambda}_{n, i}\right)\nabla^2_{\lambda}q_{\lambda^{\dagger}(\tilde{y}_{i+1})}(\theta_2)\left(\hat{\lambda}_{n, i+1}(\tilde{y}_{i+1}) - \hat{\lambda}_{n, i}\right)d\theta_2\\
    &\quad\cdot\int_{\Theta_1}p(\tilde{Y}_{i+1} = \tilde{y}_{i+1};\theta_1)q_{\hat{\lambda}_{n, i}}(\theta_1)d\theta_1d\tilde{y}_{i+1}\\
    =& \tilde{p}_{\hat{\lambda}_{n, i}}(\tilde{y}_{i+1} = y) + A(y) + B(y),
\end{align*}
where 
\begin{align*}
    A(y) &:= \left(\int_{\Theta_2}p(\tilde{Y}_{i+2} = y;\theta_2)\nabla_{\lambda}q_{\hat{\lambda}_{n, i}}(\theta_2)d\theta_2\right)^\top\mathbb{E}\left[\hat{\lambda}_{n, i+1}(\tilde{y}_{i+1}) - \hat{\lambda}_{n, i}\right]\\
    B(y) &:= \frac{1}{2}\int_{\mathcal{Y}_1}\int_{\Theta_2}p(\tilde{Y}_{i+2} = y;\theta_2)\left(\hat{\lambda}_{n, i+1}(\tilde{y}_{i+1}) - \hat{\lambda}_{n, i}\right)\nabla^2_{\lambda}q_{\lambda^{\dagger}(\tilde{y}_{i+1})}(\theta_2)\left(\hat{\lambda}_{n, i+1}(\tilde{y}_{i+1}) - \hat{\lambda}_{n, i}\right)d\theta_2\\
    &\qquad\cdot\int_{\Theta_1}p(\tilde{Y}_{i+1} = \tilde{y}_{i+1};\theta_1)q_{\hat{\lambda}_{n, i}}(\theta_1)d\theta_1d\tilde{y}_{i+1}
\end{align*}
Therefore
\begin{align*}
\xi_i &= \frac{1}{2}\int_{\mathcal{Y}_2} \left|A(y) + B(y)\right|dy\\
&\le \frac{1}{2}\int_{\mathcal{Y}_2} \left|A(y)\right|dy + \frac{1}{2}\int_{\mathcal{Y}_2} \left|B(y)\right|dy,
\end{align*}

with
\begin{align*}
   \frac{1}{2}\int_{\mathcal{Y}_2} \left|A(y)\right|dy &= \frac{1}{2} \int_{\mathcal{Y}_2}\left|\left(\int_{\Theta_2}p(\tilde{Y}_{i+2} = y;\theta_2)\nabla_{\lambda}q_{\hat{\lambda}_{n, i}}(\theta_2)d\theta_2\right)^\top\mathbb{E}\left[\hat{\lambda}_{n, i+1}(\tilde{y}_{i+1}) - \hat{\lambda}_{n, i}\right] \right|dy\\
   &\leq \frac{1}{2} \left(\int_{\mathcal{Y}_2}\left|\int_{\Theta_2}p(\tilde{Y}_{i+2} = y;\theta_2)\nabla_{\lambda}q_{\hat{\lambda}_{n, i}}(\theta_2)d\theta_2\right|dy\right)^\top\left|\mathbb{E}\left[\hat{\lambda}_{n, i+1}(\tilde{y}_{i+1}) - \hat{\lambda}_{n, i}\right] \right|\\
   &= \frac{1}{2} \left(\int_{\mathcal{Y}_2}\left|\nabla_{\lambda}\tilde{p}_{\hat{\lambda}_{n, i}}(\tilde{Y}_{i+2} = y)\right|dy\right)^\top\left|\mathbb{E}\left[\hat{\lambda}_{n, i+1}(\tilde{y}_{i+1}) - \hat{\lambda}_{n, i}\right] \right|,
\end{align*}
assuming that it is possible to change the order of the integration over $\Theta$ and the derivative in $\lambda$,

and
\begin{align*}
   &\frac{1}{2}\int_{\mathcal{Y}_2} \left|B(y)\right|dy\\
   =& \frac{1}{4} \int_{\mathcal{Y}_2}\left|\int_{\mathcal{Y}_1}\int_{\Theta_2}p(\tilde{Y}_{i+2} = y;\theta_2)\left(\hat{\lambda}_{n, i+1}(\tilde{y}_{i+1}) - \hat{\lambda}_{n, i}\right)^\top\nabla^2_{\lambda}q_{\lambda^{\dagger}(\tilde{y}_{i+1})}(\theta_2)\left(\hat{\lambda}_{n, i+1}(\tilde{y}_{i+1}) - \hat{\lambda}_{n, i}\right)d\theta_2\right.\\
   &\qquad\left.\tilde{p}_{\hat{\lambda}_{n, i}}(\tilde{Y}_{i+1} = \tilde{y}_{i+1})d\tilde{y}_{i+1} \right|dy\\
   =& \frac{1}{4} \int_{\mathcal{Y}_2}\left|\int_{\mathcal{Y}_1}\left(\hat{\lambda}_{n, i+1}(\tilde{y}_{i+1}) - \hat{\lambda}_{n, i}\right)^\top\nabla^2_{\lambda}\tilde{p}_{\lambda^{\dagger}(\tilde{y}_{i+1})}(y)\left(\hat{\lambda}_{n, i+1}(\tilde{y}_{i+1}) - \hat{\lambda}_{n, i}\right)\tilde{p}_{\hat{\lambda}_{n, i}}(\tilde{Y}_{i+1} = \tilde{y}_{i+1})d\tilde{y}_{i+1} \right|dy\\
   \le& \frac{1}{4} \int_{\mathcal{Y}_1}\int_{\mathcal{Y}_2}\left|\kappa_{\lambda^{\dagger}(\tilde{y}_{i+1})}(y)\right|\left\|\hat{\lambda}_{n, i+1}(\tilde{y}_{i+1}) - \hat{\lambda}_{n, i}\right\|_2^2\tilde{p}_{\hat{\lambda}_{n, i}}(\tilde{Y}_{i+1} = \tilde{y}_{i+1})dyd\tilde{y}_{i+1} 
\end{align*}   
where $\kappa_{\lambda^{\dagger}(\tilde{y}_{i+1})}(y)$ is the eigenvalue of $\nabla^2_{\lambda}\tilde{p}_{\lambda^{\dagger}(\tilde{y}_{i+1})}(y)$ with the largest absolute value. We can then use the Cauchy--Schwarz after writing
\begin{align*}
   & \frac{1}{4} \int_{\mathcal{Y}_1}\int_{\mathcal{Y}_2}\left|\kappa_{\lambda^{\dagger}(\tilde{y}_{i+1})}(y)\right|\left\|\hat{\lambda}_{n, i+1}(\tilde{y}_{i+1}) - \hat{\lambda}_{n, i}\right\|_2^2\tilde{p}_{\hat{\lambda}_{n, i}}(\tilde{Y}_{i+1} = \tilde{y}_{i+1})dyd\tilde{y}_{i+1} \\
   =& \frac{1}{4} \int_{\mathcal{Y}_1}\int_{\mathcal{Y}_2}\left|\kappa_{\lambda^{\dagger}(\tilde{y}_{i+1})}(y)\right|dy\sqrt{\tilde{p}_{\hat{\lambda}_{n, i}}(\tilde{Y}_{i+1} = \tilde{y}_{i+1})}\left\|\hat{\lambda}_{n, i+1}(\tilde{y}_{i+1}) - \hat{\lambda}_{n, i}\right\|_2^2\\
   &\qquad\sqrt{\tilde{p}_{\hat{\lambda}_{n, i}}(\tilde{Y}_{i+1} = \tilde{y}_{i+1})}d\tilde{y}_{i+1} \\
   \leq& \frac{1}{4} \sqrt{\int_{\mathcal{Y}_1}\left(\int_{\mathcal{Y}_2}\left|\kappa_{\lambda^{\dagger}(\tilde{y}_{i+1})}(y)\right|dy\right)^2\tilde{p}_{\hat{\lambda}_{n, i}}(\tilde{Y}_{i+1} = \tilde{y}_{i+1})d\tilde{y}_{i+1}}\\
   &\qquad\sqrt{\int_{\mathcal{Y}_1}\left\|\hat{\lambda}_{n, i+1}(\tilde{y}_{i+1}) - \hat{\lambda}_{n, i}\right\|_2^4\tilde{p}_{\hat{\lambda}_{n, i}}(\tilde{Y}_{i+1} = \tilde{y}_{i+1})d\tilde{y}_{i+1}}\\
   =& \frac{1}{4} \sqrt{\mathbb{E}\left[\left(\int_{\mathcal{Y}_2}\left|\kappa_{\lambda^{\dagger}(\tilde{y}_{i+1})}(y)\right|dy\right)^2\right]}\sqrt{\mathbb{E}\left[\left\|\hat{\lambda}_{n, i+1}(\tilde{y}_{i+1}) - \hat{\lambda}_{n, i}\right\|_2^4\right]}.
\end{align*} 
This concludes the proof.
\end{proof}

\subsection{Proof of \cref{cor:ACID_conditions}}\label{app:proof_cor_ACID_conds}

\begin{proof}
    Following \citet{battiston2025bayesian}, $\sum_{i=1}^\infty \xi_i < \infty$ almost surely is sufficient for the existence of $\Pi_{\infty}(\theta\mid y_{1:n})$. 
    Under the conditions of \cref{cor:ACID_conditions},
\begin{align*}
    \xi_i \leq &\frac{1}{2} \left(\int_{\mathcal{Y}}\left|\nabla_{\lambda}\tilde{p}_{\hat{\lambda}_{n, i}}(y)\right|dy\right)^\top\left|\mathbb{E}\left[\hat{\lambda}_{n, i+1}(\tilde{Y}_{i+1})\right] - \hat{\lambda}_{n, i} \right|\\
    &\qquad+ \frac{1}{4} \sqrt{\mathbb{E}\left[\left(\int_{\mathcal{Y}}\left|\kappa_{\lambda^{\dagger}(\tilde{Y}_{i+1})}(y)\right|dy\right)^2\right]}\sqrt{\mathbb{E}\left[\left\|\hat{\lambda}_{n, i+1}(\tilde{Y}_{i+1}) - \hat{\lambda}_{n, i}\right\|_2^4\right]}\\
    =&O(1/i^{1+\epsilon}) + O(1/i^{1+\epsilon})\\
    =&O(1/i^{1+\epsilon}),
\end{align*}
almost surely for some $\epsilon > 0$, which implies that $\sum_{i=1}^\infty \xi_i < \infty$ almost surely.
This concludes the proof.
\end{proof}

\subsection{Verification of \cref{cor:ACID_conditions} for multivariate Gaussian location model}
\label{sec:verification_acid}

We verify the conditions stipulated by \cref{cor:ACID_conditions} for $\Pi_{\infty}(\theta\mid y_{1:n})$ to be well defined for the multivariate Gaussian location model.

Firstly, recall that for any $y_{1:n}$ and $\tilde{y}_{1:i}$ the Bayesian posterior is given by $\pi(\theta \mid  y_{1:n}, \tilde{y}_{1:i}) = \mathcal{N}(\mu_{n, i}, \Sigma_{n, i})$, with
\begin{equation*}
    \Sigma_{n, i} = (I + (n+i)B)^{-1},\qquad
    \mu_{n, i} = \Sigma_{n, i} B \biggl(\sum_{j=1}^n y_j + \sum_{j=1}^i \tilde{y}_j\biggr),
\end{equation*}
where $B = A^{-1}$, and the Gaussian MF variational parameters $\hat{\lambda}_{n, i} = \{m_{n, i}, s_{n, i}\}\in \mathbb{R}^p \times \mathbb{R}^p_{> 0}$, approximating $\pi(\theta \mid  y_{1:n}, \tilde{y}_{1:i})$, are given by
\begin{align*}
    m_{n, i} = \mu_{n, i}, \qquad s_{n, i} = \mathrm{diag}(\Sigma_{n, i}^{-1})^{-1}.
\end{align*}

For the mean parameter we have that 
\begin{align*}
    m_{n, i+1} &= \Sigma_{n, i+1} B \left(\sum_{j=1}^n y_j + \sum_{j=1}^i \tilde{y}_j + \tilde{Y}_{i+1}\right) \\
    &= \Sigma_{n, i+1} \left(\Sigma_{n, i}^{-1} m_{n, i} + B \tilde{Y}_{i+1}\right) \\
    &= \Sigma_{n, i+1} \left((\Sigma_{n, i+1}^{-1} - B) m_{n, i} + B \tilde{Y}_{i+1}\right) \\
    &= m_{n, i} + \Sigma_{n, i+1} B (\tilde{Y}_{i+1} - m_{n, i}) 
\end{align*}
where $\Sigma_{n, i+1} = O(1/i)$ and $\tilde{Y}_{i+1} \sim \mathcal{N}\left(m_{n, i}, A + S_{n, i}\right)$ and therefore $m_{n, i+1} - m_{n, i}\sim \mathcal{N}\left(0, \Sigma_{n, i+1}B\left(A + S_{n, i}\right) \Sigma_{n, i+1}B\right)$.

For the scale parameter,
\begin{align*}
    s_{n, i} = \mathrm{diag}(\Sigma_{n, i}^{-1})^{-1},
\end{align*}
with
\begin{equation*}
    \Sigma_{n, i} = (I + (n+i)B)^{-1}
\end{equation*}
and note that $\Sigma_{n, i}$ is deterministic. 
Therefore
\begin{align*}
    \{s_{n, i+1}\}_{j} - \{s_{n, i}\}_{j} 
&= \frac{1}{(I + (n+i+1)B)_{jj}} - \frac{1}{(I + (n+i)B)_{jj}}\\
    &= \frac{-B_{jj}}{(I + (n+i+1)B)_{jj}(I + (n+i)B)_{jj}}\\
    &= O\left(1/i^2\right).
\end{align*}
Further, $\Sigma_{n, i}B\left(A + S_{n, i-1}\right) \Sigma_{n, i}B = O(1/i^2)$.

As a result
\begin{align*}    \left|\mathbb{E}\left[\{s_{n, i+1}\}_{j}\right] - \{s_{n, i}\}_{j}\right| &= O\left(1/i^2\right)\\
\sqrt{\mathbb{E}\left[\left\|\{s_{n, i+1}\}_{j} - \{s_{n, i}\}_{j}\right\|_2^4\right]} &= O\left(1/i^4\right)\\
\left|\mathbb{E}\left[m_{n, i+1}\right] - m_{n, i}\right| &= 0\\
\sqrt{\mathbb{E}\left[\left\|m_{n, i+1} - m_{n, i}\right\|_2^4\right]} &= O(1/i^2)
\end{align*}
almost surely.

Further, the variational predictive density takes the closed form
\begin{equation*}
\tilde p_{\hat\lambda_{n,i}}(y) = \mathcal{N}(y;\, m_{n,i},\, A + S_{n,i}),
\qquad S_{n,i} = \mathrm{diag}(s_{n,i}).
\end{equation*}
Although $s_{n,i} = O(1/i)$ as $i\to\infty$, the predictive covariance $A + S_{n,i} = O(1)$. 

Therefore, 
\begin{align*}
\nabla_{m_j}\tilde p_\lambda(y) &= \tilde p_\lambda(y)\cdot \bigl[(A+S)^{-1}(y-m)\bigr]_j,\\
\nabla_{s_j}\tilde p_\lambda(y) &= \tilde p_\lambda(y)\cdot \tfrac{1}{2}\Bigl(\bigl[(A+S)^{-1}(y-m)\bigr]_j^{\,2} - \bigl[(A+S)^{-1}\bigr]_{jj}\Bigr).
\end{align*}
Taking absolute values and integrating with respect to $y$ gives
\begin{align*}
\int_{\mathcal{Y}}\bigl|\nabla_{m_j}\tilde p_\lambda(y)\bigr|\,dy 
&= \mathbb{E}_{Y\sim\tilde p_\lambda}\bigl|[(A+S)^{-1}(Y-m)]_j\bigr|
= \sqrt{\tfrac{2}{\pi}\bigl[(A+S)^{-1}\bigr]_{jj}} = O(1),\\
\int_{\mathcal{Y}}\bigl|\nabla_{s_j}\tilde p_\lambda(y)\bigr|\,dy
&= \tfrac{1}{2}\,\mathbb{E}_{Y\sim\tilde p_\lambda}\bigl|[(A+S)^{-1}(Y-m)]_j^{\,2} - [(A+S)^{-1}]_{jj}\bigr| = O(1).
\end{align*}

Combining with the above bounds gives
\begin{align*}
\biggl(\int_{\mathcal{Y}}\bigl|\nabla_m\tilde p_{\hat\lambda_{n,i}}(y)\bigr|\,dy\biggr)^{\!\top}\bigl|\mathbb{E}_i[\Delta m_{n,i}]\bigr| &= O(1)\cdot 0 = 0,\\
\biggl(\int_{\mathcal{Y}}\bigl|\nabla_s\tilde p_{\hat\lambda_{n,i}}(y)\bigr|\,dy\biggr)^{\!\top}\bigl|\mathbb{E}_i[\Delta s_{n,i}]\bigr| &= O(1)\cdot O(1/i^2) = O(1/i^2).
\end{align*}
Both are $O(1/i^{1+\epsilon})$ for any $\epsilon\in(0,1)$, satisfying the first condition of \cref{cor:ACID_conditions}.

A similar argument can also be used to show that $\int_{\mathcal{Y}}\left|\kappa_{\lambda^{\dagger}(\tilde{Y}_{i+1})}(y)\right|dy = O(1)$. As a result, both conditions of \cref{cor:ACID_conditions} hold with $\epsilon\in(0,1)$, and $\Pi_\infty$ is well defined for the multivariate Gaussian location model.

\subsection{Regularity conditions and discussion}\label{app:reg_conds}

\begin{assumption}\label{ass:pred_dist}
The sequence $(\tilde{y}_{1:N})_{N\ge 1}$, defined by the one-step-ahead predictive density $f(y_{i+1}\mid y_{1:n}, \tilde{y}_{1:i})$, is almost conditionally identically distributed (a.c.i.d.). Further, the sequence $(\tilde{y}_{1:N})_{N\ge 1}$ is ergodic (stationarity is implied by a.c.i.d.) conditional on $y_{1:n}$ with stationary distribution ${G}_n(\cdot)$. 
\end{assumption}
\begin{remark}
This assumption is verifiable in cases where the predictive has a closed-form, such as in the case of the simple Gaussian mean. However, in general, numerical exploration of this condition is required. 
\end{remark}
Define $\ell_{N}(\theta):=\frac{1}{N}\sum_{j=1}^{N}\log p(\tilde{y}_j;\theta)$, $\ell_n(\theta)=\frac{1}{n}\sum_{j=1}^{n}\log p({y}_j;\theta)$ and $\ell_{n,N}(\theta)=\alpha_N\ell_n(\theta)+(1-\alpha_N)\ell_{N}(\theta)$, where $\alpha_N=n/N$. Let $\mathbb{E}_{n,N}$ denote the expectation with respect to the distribution of the sequence $\tilde{y}_{1},\dots,\tilde{y}_{N}$ for $y_{1:n}$ fixed, i.e., $\tilde{f}(y_{i+1}\mid y_{1:n}, \tilde{y}_{1:i})$. Define $\ell_{n,\infty}(\theta):=\lim_{N\rightarrow\infty}\ell_{n,N}(\theta)$, and let $$\theta_{n,\infty}=\theta(\{y_{1:n},\tilde{y}_{1:\infty}\}):=\argmax_{\theta\in\Theta}\ell_{n,\infty}(\theta)$$ when these quantities exist. By the ergodic nature of the sequence being generated (\cref{ass:pred_dist}), $$\lim_{N\rightarrow\infty}\ell_{n,N}(\theta)=\ell_{n,\infty}(\theta)=\mathbb{E}_{Y\sim G_n}\log p(Y\mid\theta).$$

\begin{remark}
The estimator $\theta_{n,N} := \arg\max_{\theta\in\Theta} \ell_{n,N}(\theta)$ analysed here weights the observed block $\ell_n$ and the imputed block $\ell_N$ as $(\alpha_N, 1-\alpha_N)$ with $\alpha_N = n/N$. This differs at finite $N$ from the path-terminal estimator $\hat\theta(\{y_{1:n}, \tilde y_{1:N}\})$ used in \cref{alg:gen_pred_resample,alg:MPVI}, which weights the two blocks per-observation, i.e., as $(n/(n+N), N/(n+N))$. The two estimators share the same limit $\theta_{n,\infty}$ under the regime $n/N = o(1)$ used throughout (\cref{sec:validity}), and the asymptotic results below apply to both.
\end{remark}

\begin{assumption}\label{ass:regularity}
The following conditions are satisfied and, when required, hold almost surely in $y_1,\dots,y_n$. 
\begin{enumerate}
    \item[(i)] For any $\varepsilon>0$ and $\theta,\theta'\in\Theta$ with $\|\theta-\theta'\|\ge\varepsilon$, if $\theta\ne\theta'$, then $p(y\mid\theta)\ne p(y\mid\theta')$ for $y\in A\subseteq\mathcal{Y}$ such that $G_n(A)>0$. Further, $\mathbb{E}_{Y\sim G_n}|\log p(Y\mid\theta)|<\infty$ for all $\theta\in\Theta$, and $\mathbb{E}_{Y\sim G_n}\log p(Y\mid\theta)$ admits a unique maximum. 
    \item[(ii)] The set $\Theta\subset\mathbb{R}^d$ is compact.
    \item[(iii)] For each $\theta\in\Theta$, $\log p(y\mid\theta)$ is continuous with $G_n$-probability 1. For all $\theta,\theta'\in\Theta$, there exists $d(y)$ such that $|\log p(y\mid\theta)-\log p(y\mid\theta')|\le d(y)\|\theta-\theta'\|$, and $\mathbb{E}_{Y\sim G_n}[d(Y)^2]<\infty$.
    \item[(iv)] The function $\theta\mapsto \ell_{n,\infty}(\theta)$ is twice continuously differentiable in a neighbourhood of $\theta_{n,\infty}$, and has a non-singular second-derivative matrix at $\theta_{n,\infty}$ denoted by $V_n$. 
    \item[(v)] $\nabla_\theta \ell_{N}(\theta_{n,\infty})=O_p(1/\sqrt{N})$ as $N\rightarrow\infty$ under $G_n$. 
 \end{enumerate}
\end{assumption}
\begin{remark}
    Assumptions (i)--(v) are standard in the frequentist analysis of point estimators, and since our draws from $\Pi_N(\cdot\mid y_{1:n})$ are obtained as $$\theta_{n,N}=\theta(\{y_{1:n},\tilde{y}_{1:N}\})=\argmax_{\theta\in\Theta} \ell_{n,N}(\theta),$$ it should not be surprising that similar conditions are required here. Most notably, however, these conditions must hold almost surely for the sequence $y_{1},\dots,y_n$. This condition is non-standard and is required since the generation of data used in producing $\theta_{n,N}$ is, by construction, conditional on $y_1,\dots,y_n$.
\end{remark}

\subsection{Proof of \cref{lem:one}}\label{app:proof_theta_conv}
\begin{proof}
Along a given sequence $\tilde{y}_{1:N}$, Assumptions~\ref{ass:regularity}(i)-(ii) imply existence of $\theta_{n,\infty}$ provided that $\ell_{n,\infty}(\theta)=\lim_{N}\ell_{n,N}(\theta)$ exists, which is guaranteed under Assumptions~\ref{ass:pred_dist} and \ref{ass:regularity}(i). The a.c.i.d.~condition in \cref{ass:pred_dist} then implies that the random variable $\theta_{n,\infty}$ exists almost surely.

From compactness and continuity of $\ell_{n,N}(\theta)$, Assumptions~\ref{ass:regularity}(ii)-(iii), the maximiser of $\ell_{n,N}(\theta)$, denoted by $\theta_{n,N}$, exists uniformly in $N$. Now, under Assumptions~\ref{ass:regularity}(i)-(iii), we have that, given $y_{1:n}$,  
$
\sup_{\theta\in\Theta}|\ell_{n,N}(\theta)-\ell_{n,\infty}(\theta)|=o_p(1),
$ (see, e.g., Example 19.7 in \citealp{van2000asymptotic} for details) and this uniform convergence implies that the maximiser $\theta_{n,N}$ also converges so long as it is unique, which we will now demonstrate.  

Let $\hat\theta_{n,N_k}$ be any convergent subsequence with limit  $\theta_{n,\infty}$. Then, for any other $\theta^\star_{n,\infty}\in\Theta$, we have that, a.s. in $y_{1:n}$,
\begin{flalign*}
\ell_{n,\infty}(\theta_{n,\infty})\ge& \lim_{k\rightarrow\infty}\left\{\sup_{\theta\in\Theta}|\ell_{n,N_k}(\theta)-\ell_{n,\infty}(\theta)|+\ell_{n,N_k}(\hat\theta_{n,N_k})\right\}\\\ge& \lim_{k\rightarrow\infty}\left\{\sup_{\theta\in\Theta}|\ell_{n,N_k}(\theta)-\ell_{n,\infty}(\theta)|+\ell_{n,N_k}(\theta^\star_{n,\infty})\right\}\\=&\ell_{n,\infty}(\theta^\star_{n,\infty})+o(1). 
\end{flalign*}
Hence, $\theta_{n,\infty}$ maximises $\ell_{n,\infty}(\theta)$. However, since by \cref{ass:regularity}(ii), maximisers of $\ell_{n,N}(\theta)$ form a bounded sequence, from continuity of $\ell_{n,N}(\theta)$, each sequence must admit a convergent subsequence. Therefore, it must be that $\ell_{n,\infty}(\theta_{n,\infty})=\ell_{n,\infty}(\theta^\star_{n,\infty})$, and by identifiability (\cref{ass:regularity}(i)) we have $\theta_{n,\infty}=\theta^\star_{n,\infty}$. Hence, $\hat\theta_{n,N}\rightarrow\theta_{n,\infty}$ a.s. in $y_{1:n}$.

Now, we note that under Assumptions~\ref{ass:regularity}(iii)-(v), using \citet[Theorem~3.2.21]{van1996weak}, we have that 
\begin{flalign*}
\sqrt{N}\left(\theta_{n,N}-\theta_{n,\infty}\right)&=-V_n^{-1}\left\{\sqrt{N}\alpha_N\nabla_\theta\ell_{n}(\theta_{n,\infty})+(1-\alpha_N)\sqrt{N}\nabla_\theta\ell_{N}(\theta_{n,\infty})\right\}+o_p(1)\\&=-(1-\alpha_N)V_n^{-1}\sqrt{N}\nabla_\theta \ell_{N}(\theta_{n,\infty})+O_P(n/\sqrt{N}).
\end{flalign*}
Under \cref{ass:regularity}(v), almost surely in $y_{1:n}$, for any $\epsilon>0$, there is an $M_\epsilon>0$ large enough so that $\Pr\left\{\|V_n^{-1}\sqrt{N}\nabla_\theta \ell_{N}(\theta_{n,\infty})\|>M_\epsilon\mid y_{1:n}\right\}< \epsilon$. Hence, we can conclude that, for any $\epsilon>0$ and $M_\epsilon$ as above
$$
\Pr\left\{\|\sqrt{N}\left(\theta_{n,N}-\theta_{n,\infty}\right)\|>M_\epsilon\mid y_{1:n}\right\}<\epsilon,
$$
and the last stated result follows. 
\end{proof}

\subsection{Proof of \cref{lem:KLD_MFVI}}\label{app:proof_lem_KLD_MFVI}

\begin{proof}
    Both posterior $\pi(\theta\mid y_{1:n})$ and variational approximation $q_{\hat{\lambda}_{n, 0}}(\theta)$ are multivariate Gaussians with the same mean $\mu_{n, 0} = m_{n, 0}$, and so the KLD reduces to
\begin{align*}
    \KL[q_{\hat{\lambda}_{n, 0}}(\theta) \parallel  \pi(\theta \mid  y_{1:n})] &= \frac{1}{2} \left(\mathrm{tr}(\Sigma_{n, 0}^{-1} S_{n, 0}) - d - \log\left(\frac{|\Sigma_{n, 0}^{-1}|}{|S_{n, 0}^{-1}|}\right)\right) \\
    &= \frac{1}{2}\left(\sum_{i} (\Sigma_{n, 0}^{-1})_{ii} (S_{n, 0})_{ii} - d - \log\left(\frac{|I + n B|}{\prod_i (1 + nB_{ii})}\right)\right) \\
    &= -\frac{1}{2} \log \left(\frac{\left|\frac{1}{n} I + B\right|}{\prod_i \left(\frac{1}{n} + B_{ii}\right)}\right).
\end{align*}
Hence, if $B$ is not diagonal $\KL[q_{\hat{\lambda}_{n, 0}}(\theta) \parallel  \pi(\theta \mid  y_{1:n})]>0$. Furthermore, as $B$ is fixed independent of $n$,
\begin{equation*}
    \lim_{n \to \infty} \KL[q_{\hat{\lambda}_{n, 0}}(\theta) \parallel  \pi(\theta \mid  y_{1:n})] = -\frac{1}{2} \log \left(\frac{|B|}{\prod_i B_{ii}}\right) > 0;
\end{equation*}
therefore when $B$ is not diagonal, the KLD does not vanish even as $n \to \infty$.
\end{proof}

\subsection{Proof of \cref{thm:KLD_MPVI}}
\label{app:proof_MPVI_gaussian_mean}

\begin{proof}

Firstly, since the MLE of a multivariate Gaussian location model is the sample mean vector of the observations, in the multivariate Gaussian location model we have that $\lim_{N\rightarrow\infty} m_{n, N} = \theta(y_{1:n}, \tilde{y}_{1:\infty})$. Therefore, we can analyse $\pi_{\infty}(\theta\mid y_{1:n})$ by investigating the asymptotic distribution of $m_{n, N}$.

\paragraph{Asymptotic distribution of $m_{n, N}$.}

Notice that the sequence of posterior covariances $(\Sigma_{n, N})_{N \geq 0}$, and hence also that of variational covariances $(S_{n, N})_{N \geq 0}$, are deterministic and converge to $0$ as $N \to \infty$.
Let $C_{n, N} = A + S_{n, N}$ and $K_{n, N} = \Sigma_{n, N} B$, which are symmetric.
The sequence of variational means $(m_{n, N})_{N \geq 0}$ satisfies the recursion
\begin{align*}
    m_{n, N} &= \Sigma_{n, N} B \left(\sum_{j=1}^n y_j + \sum_{j=1}^{N-1} \tilde{y}_j + \tilde{Y}_N\right) \\
    &= \Sigma_{n, N} \left(\Sigma_{n, N-1}^{-1} m_{n, N-1} + B \tilde{Y}_N\right) \\
    &= \Sigma_{n, N} \left((\Sigma_{n, N}^{-1} - B) m_{n, N-1} + B \tilde{Y}_N\right) \\
    &= m_{n, N-1} + \Sigma_{n, N} B (\tilde{Y}_N - m_{n, N-1}) \\
    &= m_{n, N-1} + K_{n, N} (\tilde{Y}_N - m_{n, N-1}),
\end{align*}
where
\begin{equation*}
    \tilde{Y}_N - m_{n, N-1} \sim \mathcal{N}(0, A + S_{n, N-1})
\end{equation*}
is independent of $\{y_{1:n}, \tilde{y}_{1:N-1}\}$. 
This demonstrates that the MF variational mean sequence is a martingale.
Similarly, the increments
\begin{equation*}
    \Delta m_{n, N} = m_{n, N} - m_{n, N-1} \sim \mathcal{N}\left(0, K_{n, N} C_{n, N-1} K_{n, N}\right),
\end{equation*}
are also independent of $\{y_{1:n}, \tilde{y}_{1:N-1}\}$.
Hence, we have that
\begin{equation*}
    m_{n, N} \mid  y_{1:n}
    = m_{n, 0} + \sum_{j=1}^N \Delta m_{n,j} \sim \mathcal{N}\left(\mu_{n, 0}, V_{n, N}\right),
\end{equation*}
where $V_{n, N} = \sum_{j=1}^N M_{n,j}$ with $M_{n, N} = K_{n, N} C_{n, N-1} K_{n, N}$.
Each $M_{n, N}$ is symmetric PSD, so that $V_{n, N}$ is increasing in Loewner order.
We now show that the series $V_{n, N}$ converges.
Expanding $M_{n, N}$, we have that
\begin{align*}
    M_{n, N} &= K_{n, N} C_{n, N-1} K_{n, N} \\
    &= \Sigma_{n, N} B (A + S_{n, N-1}) B \Sigma_{n, N} \\
    &= \Sigma_{n, N} B \Sigma_{n, N} + K_{n, N} S_{n, N-1} K_{n, N}.
\end{align*}
Diagonalise $B$ as $B = U \mathrm{diag}(\lambda_i(B)) U^T$, where $U$ is orthogonal and $\lambda_i(B) > 0$ are the eigenvalues of $B$.
Then,
\begin{align*}
    \Sigma_{n, N} B \Sigma_{n, N} &= (I + (n+N) B)^{-1} B (I + (n+N) B)^{-1} \\
    &= U (I + (n+N) \mathrm{diag}(\lambda_i(B)))^{-1} \mathrm{diag}(\lambda_i(B)) (I + (n+N) \mathrm{diag}(\lambda_i(B)))^{-1} U^T \\
    &= U \mathrm{diag}\left(\frac{\lambda_i(B)}{(1 + (n+N) \lambda_i(B))^2}\right) U^T \\
    &= O\left((n+N)^{-2}\right)
\end{align*}
in spectral norm, as its eigenvalues satisfy
\begin{equation*}
    \frac{\lambda_i(B)}{(1 + (n+N) \lambda_i(B))^2} \leq \frac{1}{(n+N)^2 \lambda_\text{min}(B)}.
\end{equation*}
Furthermore,
\begin{align*}
    \lambda_\text{max}(K_{n, N} S_{n, N-1} K_{n, N}) &= \lambda_\text{max}(\Sigma_{n, N} B \mathrm{diag}(\Sigma_{n, N-1}^{-1})^{-1} B \Sigma_{n, N}) \\
    &\leq \lambda_\text{max}(\Sigma_{n, N} B)^{2} \lambda_\text{max}(\mathrm{diag}(I + (n+N-1) B)^{-1}) \\
    &= \left(\max_i\left\{\frac{\lambda_i(B)}{1 + (n+N)\lambda_i(B)}\right\}\right)^2 \max_i \left\{(1 + (n+N-1) B_{ii})^{-1}\right\} \\
    &\leq \frac{1}{(n+N)^2} \cdot \frac{1}{1 + (n+N-1) \lambda_\text{min}(B)},
\end{align*}
which implies that
\begin{equation*}
    K_{n, N} S_{n, N-1} K_{n, N} = O\left((n+N)^{-3}\right)
\end{equation*}
in spectral norm.
Combining the above, we have that
\begin{equation*}
    M_{n, N} = O\left((n+N)^{-2}\right)
\end{equation*}
in spectral norm, and the series converges absolutely to $V_{n,\infty} = \sum_{j=1}^\infty M_{n,j}$.

Moreover, $m_{n, N} \mid  y_{1:n}$ is square-integrable, since
\begin{equation*}
    \sup_N \mathbb{E}\left[\|m_{n, N}\|^2 \mid y_{1:n}\right] = \|\mu_{n, 0}\|^2 + \mathrm{tr}(V_{n,\infty}) < \infty.
\end{equation*}
Hence, the sequence $(m_{n, N})_{N \geq 0}$ is a square-integrable martingale with respect to the filtration $(\mathcal{F}_{n, N})_{N \geq 0}$, and converges almost surely and in $L^2$ to $m_{n,\infty}$ as $N \to \infty$ by the martingale convergence theorem, where
\begin{equation*}
    m_{n,\infty} \mid  \mathcal{F}_{n, 0} \sim \mathcal{N}(\mu_{n, 0}, V_{n,\infty}).
\end{equation*}

\paragraph{Comparison to true posterior}
Notice that the deterministic decrement of the posterior covariance can be written as
\begin{align*}
    \Sigma_{n, N-1} - \Sigma_{n, N} &= \Sigma_{n, N-1} - (I + (n+N) B)^{-1} \\
    &= \Sigma_{n, N-1} - \left(\Sigma_{n, N-1}^{-1} + B\right)^{-1} \\
    &= \Sigma_{n, N-1} (B^{-1} + \Sigma_{n, N-1})^{-1} \Sigma_{n, N-1} \\
    &= \Sigma_{n, N-1} (I + B\Sigma_{n, N-1})^{-1} B \Sigma_{n, N-1} \\
    &= \Sigma_{n, N} B \Sigma_{n, N-1} \\
    &= K_{n, N} \Sigma_{n, N-1} (K_{n, N})^{-1} K_{n, N} \\
    &= K_{n, N} \Sigma_{n, N-1} (\Sigma_{n, N-1}^{-1} + B) A K_{n, N} \\
    &= K_{n, N} (A + \Sigma_{n, N-1}) K_{n, N}.
\end{align*}
Adding and subtracting $\Sigma_{n, N-1} - \Sigma_{n, N}$ in $M_{n, N}$, we have that
\begin{align*}
    M_{n, N} &= (\Sigma_{n, N-1} - \Sigma_{n, N}) + K_{n, N} (C_{n, N-1} - A - \Sigma_{n, N-1}) K_{n, N} \\
    &= (\Sigma_{n, N-1} - \Sigma_{n, N}) + K_{n, N} (S_{n, N-1} - \Sigma_{n, N-1}) K_{n, N} \\
    &= (\Sigma_{n, N-1} - \Sigma_{n, N}) + R_{n, N},
\end{align*}
where $R_{n, N}$ represents a residual term.

Notice that $S_{n, N-1} - \Sigma_{n, N-1}$ is not necessarily PSD.
As seen above, $K_{n, N} S_{n, N-1} K_{n, N} = O((n+N)^{-3})$.
Similarly, we have that
\begin{align*}
    \lambda_\text{max}(K_{n, N} \Sigma_{n, N-1} K_{n, N}) &= \lambda_\text{max}(\Sigma_{n, N} B \Sigma_{n, N-1} B \Sigma_{n, N}) \\
    &\leq \lambda_\text{max}(\Sigma_{n, N} B)^{2} \lambda_\text{max}((I + (n+N-1) B)^{-1}) \\
    &= \left(\max_i\left\{\frac{\lambda_i(B)}{1 + (n+N)\lambda_i(B)}\right\}\right)^2 \max_i \left\{(1 + (n+N-1) \lambda_i(B))^{-1}\right\} \\
    &\leq \frac{1}{(n+N)^2} \cdot \frac{1}{1 + (n+N-1) \lambda_\text{min}(B)},
\end{align*}
which also implies that $K_{n, N} \Sigma_{n, N-1} K_{n, N} = O\left((n+N)^{-3}\right)$ in spectral norm.
Overall, we have that
\begin{equation*}
    R_{n, N} = O\left((n+N)^{-3}\right)
\end{equation*}
in spectral norm.

Then, the covariance matrix at step $N$ can be expressed as
\begin{align*}
    V_{n, N} &= \sum_{j=1}^N M_{n,j} \\
    &= \sum_{j=1}^N (\Sigma_{n,j-1} - \Sigma_{n,j}) + \sum_{j=1}^N R_{n,j} \\
    &= \Sigma_{n, 0} - \Sigma_{n, N} + \sum_{j=1}^N R_{n,j} \\
    &= \Sigma_{n, 0} - \Sigma_{n, N} + E_{n, N},
\end{align*}
and at convergence we have
\begin{align*}
    V_{n, \infty} &= \Sigma_{n, 0} + \sum_{j=1}^\infty R_{n,j} \\
    &= \Sigma_{n, 0} + E_{n, \infty},
\end{align*}
where $E_{n, \infty} = O(n^{-2})$ in spectral norm.

\paragraph{Evaluating the KLD}

Then both the Bayesian posterior \eqref{eq:bayes-posterior} and the VPR posterior $\pi_{\infty}(\theta\mid y_{1:n})$ are multivariate Gaussians with the same mean $\mu_{n, 0} = m_{n, 0}$, but different covariance matrices, and so the KLD reduces to
\begin{align*}
    \KL[\pi_{\infty}(\cdot\mid y_{1:n}) \parallel  \pi(\cdot \mid  y_{1:n})] &= \frac{1}{2} \left(\mathrm{tr}(\Sigma_{n, 0}^{-1} V_{n,\infty}) - d - \log\left(\frac{|\Sigma_{n, 0}^{-1}|}{|V_{n,\infty}^{-1}|}\right)\right) \\
    &= \frac{1}{2} \left(\mathrm{tr}(I + \Sigma_{n, 0}^{-1} E_{n,\infty}) - d - \log\left(\frac{|\Sigma_{n, 0}^{-1}|}{|(\Sigma_{n, 0} + E_{n,\infty})^{-1}|}\right)\right) \\
    &= \frac{1}{2} \left(\mathrm{tr}(\Sigma_{n, 0}^{-1} E_{n,\infty}) - \log\left(|I + \Sigma_{n, 0}^{-1} E_{n,\infty}|\right)\right) \geq 0,
\end{align*}
with equality if and only if $E_{n,\infty} = 0$, which happens only when $B$ is diagonal.
However, using the fact that $E_{n,\infty} = O(n^{-2})$ in spectral norm and $\Sigma_{n, 0}^{-1} = I + n B = O(n)$ in the spectral norm, we conclude that
\begin{equation*}
    \lim_{n \to \infty}\KL[\pi_{\infty}(\cdot\mid y_{1:n}) \parallel  \pi(\cdot \mid  y_{1:n})] =0
\end{equation*}
at rate $O(n^{-2})$, completing the proof.
\end{proof}

The left-hand panel of \cref{fig:KLD_VI_vs_MP} presents a jellyfish plot of trajectories for $\{m_{n,i}^{(l)}\}_1$, the MF location parameter associated with the first dimension of $\theta$, across paths $l=1,\ldots,L$ as $i$ increases from $0$ to $N=2000$ for the Gaussian location model with $n=10$ observations.
After about $1500$ forward simulations, the trajectories appear to have stabilised.
The right-hand panels of \cref{fig:KLD_VI_vs_MP} plot $\KL[q_{\hat{\lambda}_{n,0}}(\theta)\parallel \pi(\cdot\mid y_{1:n})]$ and $\KL[\pi_{\infty}(\theta\mid y_{1:n})\parallel \pi(\cdot\mid y_{1:n})]$ as $n$ increases.
These figures illustrate that the MF-VI KL divergence approaches a non-zero limit as $n$ grows, whereas the VPR KL divergence decreases rapidly with $n$ and is much smaller even for small $n$.

\begin{figure}[!ht]
\centering
\includegraphics[width=1.0\textwidth]{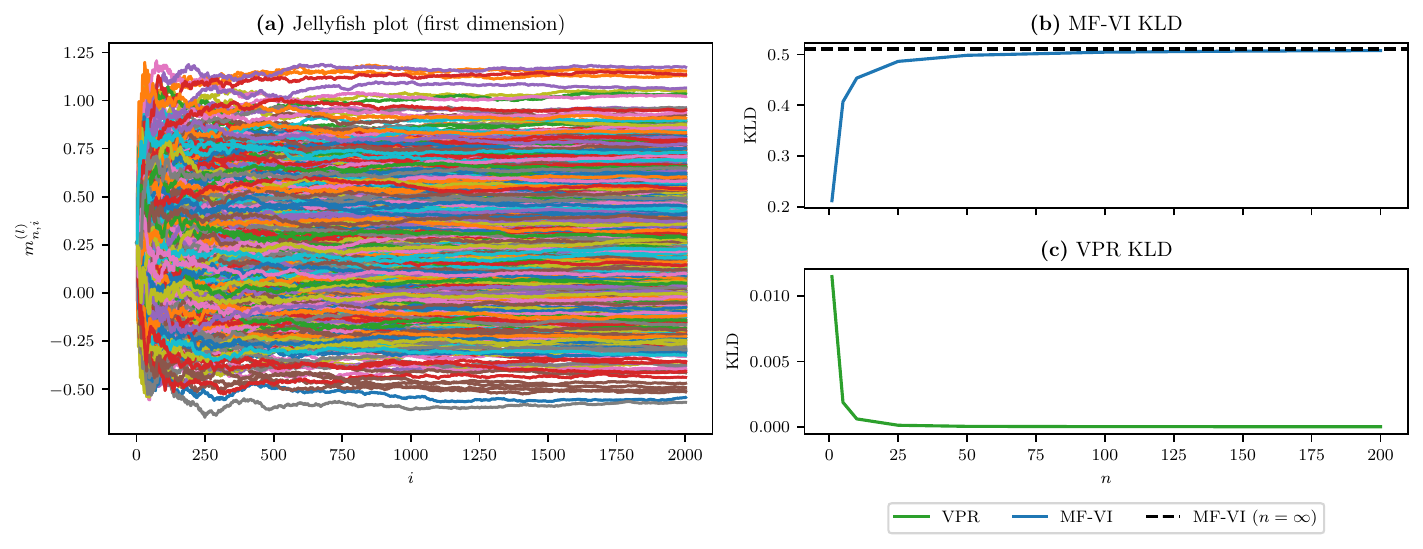}
\caption{(a): Jellyfish plot of $L = 1000$ trajectories for the first dimension of the variational mean, $\{m_{n, i}^{(l)}\}_1$, when $n = 10$; (b)/(c): KLDs of the optimal variational posterior and the VPR posterior (\cref{alg:MPVI}), respectively, to the Bayesian posterior as a function of $n$.}
\label{fig:KLD_VI_vs_MP}
\end{figure}

\section{Scalable variational predictive resampling}
\label{app:scalable-vpr}

\Cref{alg:MPVI} in \cref{sec:vpr} describes the ideal VPR procedure: at every resampling step the variational parameters are re-optimised to convergence on the augmented dataset, and the parameter associated with each path is obtained by solving a final maximum-likelihood problem on $\{y_{1:n}, \tilde y_{1:N}^{(l)}\}$. Both operations are prohibitive when repeated $N$ times across $L$ paths, especially in non-conjugate models where the ELBO maximiser is unavailable in closed form. As discussed in \cref{sec:scalability}, we therefore replace each step with a cheap approximation:
\begin{enumerate}
    \item[(i)] we approximate the variational update $\hat\lambda_{n,i+1} = \arg\max_\lambda \mathcal{L}(\lambda; y_{1:n}, \tilde y_{1:i+1})$ by $S$ minibatch stochastic-gradient steps warm-started at $\hat\lambda_{n,i}$;
    \item[(ii)] we approximate the path-terminal estimator $\theta_{n,N}^{(l)} = \hat\theta(\{y_{1:n}, \tilde y_{1:N}^{(l)}\})$ by the terminal variational mean $m_{n,N}^{(l)}$, which under standard regularity conditions is a consistent and asymptotically normal estimator of $\theta_{n,N}^{(l)}$ \citep{ray2022variational, zhang2020convergence}.
\end{enumerate}
The resulting scheme is given in \cref{alg:scale_MPVI}. It takes as additional inputs the number of gradient steps per resampling step $S$, the stepsize $\eta$, and the minibatch size $b$.

\begin{algorithm}
\caption{Scalable variational predictive resampling}
\begin{algorithmic}[1]
\STATE {\bfseries Input:} Depth $N$, $L$ samples, data $y_{1:n}$, gradient steps $S$, stepsize $\eta$, batch size $b$.
\STATE Compute $\hat{\lambda}_{n, 0} \gets \argmax_\lambda \mathcal L(\lambda; y_{1:n})$
\FOR{$l = 1$ {\bfseries to} $L$}
    \STATE Set $\hat{\lambda}^{(l)}_{n, 0} \gets \hat{\lambda}_{n, 0}$
    \FOR{$i =0$ {\bfseries to} $N - 1$}
        \STATE Sample $\tilde{y}^{(l)}_{i+1} \sim \tilde{p}_{\hat{\lambda}_{n, i}^{(l)}}(\cdot)$ and set $\hat{\lambda}^{(l)}_{n, i + 1} \gets \hat{\lambda}^{(l)}_{n, i}$
        \FOR{$s =1$ {\bfseries to} $S$}
        \STATE Draw minibatch $\mathcal{B}_{i,s} \subset \{y_{1:n}, \tilde{y}^{(l)}_{1:i+1}\}$ of size $b$
        \STATE Set  $\hat{\lambda}^{(l)}_{n,i+1} \gets \hat{\lambda}^{(l)}_{n,i+1} + \eta\,\nabla_\lambda \mathcal L(\hat{\lambda}^{(l)}_{n,i+1}; \mathcal{B}_{i,s})$
        \ENDFOR
    \ENDFOR
    \STATE Form $\theta_{n, N}^{(l)} \gets m^{(l)}_{n, N}$ 
\ENDFOR
\STATE {\bfseries Return:} $\{\theta_{n, N}^{(l)}\}_{l=1}^L$, approximating samples from $\pi(\theta\mid y_{1:n})$.
\end{algorithmic}
\label{alg:scale_MPVI}
\end{algorithm}

\section{Experimental details}\label{app:experimental_details}
\subsection{Implementation}

All experiments are implemented in JAX \citep{jax2018github} for automatic differentiation, vectorisation, and compilation via \texttt{jit}. JAX is an open-source numerical computing library that extends NumPy functionality with automatic differentiation and GPU/TPU support, designed for high-performance machine learning research. Model and variational-posterior components are written using Flax \citep{flax2020github}, and distributional primitives and sampling utilities use Distrax \citep{deepmind2020jax}. Randomness is controlled throughout with explicit JAX PRNG keys (a fixed base seed with replicate-specific offsets) to ensure reproducibility across simulated datasets and inference runs. Variational objectives and predictive-resampling updates are optimised with Adam \citep{kingma2014adam} using full-data or minibatch gradients as specified per experiment. Reference MCMC posterior samples are obtained with NumPyro \citep{phan2019composable} using the NUTS sampler \citep{hoffman2014no}: a modern, JAX-native, state-of-the-art gradient-based MCMC implementation that benefits from the same compilation and accelerator support as the rest of our pipeline; we run an initial warmup and retained-draw budget and extend sampling until a minimum effective sample size threshold is met, monitoring standard diagnostics (ESS and $\hat R$). All experiments were run on a Google Cloud Virtual Machine equipped with a v2 TPU with 8 cores, kindly provided by Google at no cost through the TPU Research Cloud (TRC) program.

\subsection{Data}\label{app:data}
\subsubsection{Linear regression}\label{app:linear_data}

We generate synthetic datasets from a Gaussian linear regression model with controlled growth in dimension and posterior anisotropy. We use 100 independent datasets for each $(n,d)$, except for $d=1000$, where we use $50$.
For each dimension $d$, we set the sample size to scale linearly as $n = 3d$. We then construct a design matrix $X \in \mathbb{R}^{n \times d}$ whose Gram matrix has a prescribed spectrum: we choose eigenvalues
\[
\lambda_1,\ldots,\lambda_d \in \bigl[\kappa(d)^{-1/2},\, \kappa(d)^{1/2}\bigr]
\]
evenly spaced on the log scale (so the geometric mean is $1$), implying $\mathrm{cond}(X^\top X)=\kappa(d)$. We realise this by drawing $U\in\mathbb{R}^{n\times d}$ with orthonormal columns and $V\in\mathbb{R}^{d\times d}$ orthogonal (random rotation), and setting
\[
X = U \,\mathrm{diag}\!\bigl(\sqrt{\lambda_1},\ldots,\sqrt{\lambda_d}\bigr)\, V^\top .
\]
The target condition number increases with dimension according to the schedule
\[
\kappa(d) \;=\; \min\!\left\{\kappa_{\max},\; \kappa_{\mathrm{ref}}\left(\frac{d}{d_{\mathrm{ref}}}\right)^{\alpha}\right\},
\qquad
\kappa_{\mathrm{ref}}=350,\;\; d_{\mathrm{ref}}=20,\;\; \alpha=1.5,\;\; \kappa_{\max}=10^{12},
\]
so posterior conditioning worsens as $d$ grows.

Given $X$, we fix a noise variance $\sigma^2$ and a ground-truth coefficient vector $\beta^\star\in\mathbb{R}^d$ obtained by drawing $\beta^\star \sim \mathcal{N}(0, I_d)$ and rescaling it so that the signal variance $\mathrm{Var}(X\beta^\star)$ matches a fixed target (held constant across $d$). Responses are generated as
\[
y = X\beta^\star + \varepsilon, \qquad \varepsilon \sim \mathcal{N}(0, \sigma^2 I_n).
\]
For each $(n,d)$ configuration we hold $(\beta^\star,\sigma^2)$ fixed and generate 100 independent datasets by resampling both $X$ (with the same spectral schedule) and the noise $\varepsilon$.

\subsubsection{Logistic regression}\label{app:logistic_data}

\paragraph{Highly-correlated feature simulation.}
We simulate binary outcomes $y\in\{0,1\}^n$ from a correlated-design logistic regression model. Concretely, we generate a design matrix $X\in\mathbb{R}^{n\times d}$ with rows $x_i^\top$ drawn i.i.d. as $x_i \sim \mathcal{N}(0,\Sigma)$, where $\Sigma\in\mathbb{R}^{d\times d}$ has Toeplitz form $\Sigma_{jk}=\rho^{|j-k|}$ to induce strong feature correlations; we set $n=15$, $d=5$, and $\rho=0.99$. We then draw a ground-truth coefficient vector $\beta^\star\in\mathbb{R}^d$ via $\beta^\star=z/\sqrt{d}$ with $z\sim\mathcal{N}(0,I_d)$, and sample labels independently as $y_i \mid x_i,\beta^\star \sim \mathrm{Bernoulli}(p_i)$ with $p_i=\sigma(x_i^\top\beta^\star)$ and $\sigma(t)=1/(1+e^{-t})$. The simulated dataset consists of the realised pair $(X,y)$ together with the generating coefficients $\beta^\star$.

\paragraph{Real datasets.}
We evaluate the non-conjugate logistic-regression experiment of \cref{sec:logreg_datasets} on five publicly available binary-classification benchmarks of varying dimension: \texttt{german}, \texttt{mozilla4}, \texttt{phoneme}, \texttt{skin}, and \texttt{telescope}. Datasets are obtained from the UCI Machine Learning Repository \citep{asuncion2007uci} (\texttt{german}, \texttt{telescope}) and OpenML \citep{vanschoren2014openml} (\texttt{mozilla4}, \texttt{phoneme}, \texttt{skin}). \cref{tab:logreg_datasets} summarises their characteristics. For each dataset and each of 100 random train/test splits, we draw a training set of size $n_{\mathrm{tr}} = 100$ uniformly at random, and use the remaining samples as a held-out test set. All continuous features are standardised to zero mean and unit variance using statistics computed on the training split.

\begin{table}[h]
\centering
\small
\caption{Characteristics of the logistic-regression benchmarks. $d$ is the number of features after standardisation; $n_{\mathrm{tr}}$ and $n_{\mathrm{te}}$ are the training- and test-set sizes used in our setup.}
\label{tab:logreg_datasets}
\begin{tabular}{lrrrl}
\toprule
Dataset & $d$ & $n_{\mathrm{tr}}$ & $n_{\mathrm{te}}$ & Source / domain \\
\midrule
\texttt{german}     & 20 & 100 & 900     & UCI; statlog credit-risk classification \\
\texttt{mozilla4}   &  5 & 100 & 15{,}445  & OpenML; defect prediction in the Mozilla codebase \\
\texttt{phoneme}    &  5 & 100 & 5{,}304   & OpenML; nasal vs.\ oral phoneme classification \\
\texttt{skin}       &  3 & 100 & 244{,}957 & OpenML; skin-segmentation pixel classification \\
\texttt{telescope}  & 10 & 100 & 18{,}920  & UCI; gamma-ray vs.\ hadron shower discrimination \\
\bottomrule
\end{tabular}
\end{table}

\subsubsection{Linear mixed-effects}\label{app:lmre_data}
We simulate clustered responses from a linear mixed random-effects model with group-specific latent effects. For groups $g=1,\dots,G$, we generate
\[
y_g = X_g \beta + Z_g u_g + \varepsilon_g,\qquad
u_g \sim \mathcal{N}(0, D),\qquad
\varepsilon_g \sim \mathcal{N}(0,\sigma^2 I_{n_g}),
\]
where $\beta\in\mathbb{R}^d$ are population-level (fixed-effect) coefficients shared across groups and $u_g\in\mathbb{R}^r$ are group-specific random effects capturing latent heterogeneity. We set $d=r=3$, $G=4$, group sizes $(n_g)=(15,25,12,30)$ (total $n=82$), and fix $\sigma^2=1$. For each dataset, we treat the design matrices $X_g\in\mathbb{R}^{n_g\times d}$ and $Z_g\in\mathbb{R}^{n_g\times r}$ as observed covariates (generated once per replicate according to the simulation design and then conditioned upon). We draw $\beta$ from a mean-zero Gaussian with a correlated covariance $\Sigma_\beta$, fix a correlated positive definite matrix $D$ for the random-effects covariance, sample $u_g \mid D \stackrel{\text{iid}}{\sim}\mathcal{N}(0,D)$, and then sample $\varepsilon_g$ independently to form $y_g$. We repeat this procedure for 100 independent simulated datasets.

\subsection{Inference set-up}\label{app:inference_setup}
\subsubsection{Linear regression}

We fit Bayesian linear regression on the simulated design--response pair $(X,y)$ described above, with Gaussian likelihood
\[
y \mid \beta \sim \mathcal{N}(X\beta,\sigma^2 I_n),
\]
and a spherical Gaussian prior
\[
\beta \sim \mathcal{N}(0,\tau_\beta^2 I_d),\qquad \tau_\beta=1000.
\]
Under this conjugate model, the exact posterior is Gaussian, $\beta\mid (X,y)\sim\mathcal{N}(m_n,S_n)$, with
\[
S_n^{-1} \;=\; \frac{1}{\sigma^2}X^\top X \;+\; \frac{1}{\tau_\beta^2}I_d,
\qquad
m_n \;=\; S_n\Bigl(\frac{1}{\sigma^2}X^\top y\Bigr).
\]

By taking a large prior scale, the conditioning of the posterior precision, and hence of $S_n$, is dominated by the conditioning of $X^\top X$.

\paragraph{MF-VI.}
We use a diagonal-covariance Gaussian variational family $q(\beta)=\prod_{j=1}^d \mathcal{N}(\beta_j\mid m_j,s_j^2)$. In this conjugate setting, the ELBO-optimal mean-field solution is available in closed form: the variational mean matches the exact posterior mean, $m_{\mathrm{MF}}=m_n$, and the variational covariance is the inverse diagonal of the exact posterior precision,
\[
S_{\mathrm{MF}}=\mathrm{diag}\!\bigl((S_n^{-1})_{11}^{-1},\ldots,(S_n^{-1})_{dd}^{-1}\bigr),
\]
so that $q(\beta)=\mathcal{N}(m_n,S_{\mathrm{MF}})$. We use this closed-form MF-VI solution to initialise VPR.

\paragraph{VPR.}
We run VPR with horizon $N=12000$ and $2000$ predictive-resampling paths, using the same diagonal-Gaussian variational parameterisation as in MF-VI. At each predictive-resampling step we (i) propagate covariates forward by bootstrap-resampling from the observed rows of $X$, using a single shared resampled covariate stream across all paths, and (ii) impute the next response by sampling from the current variational predictive distribution under the Gaussian model. We then update the variational parameters using the optimal MF-VI update for linear regression: after augmenting the dataset with the newly imputed observation, we recompute the closed-form mean-field solution (posterior mean with covariance equal to the inverse diagonal of the posterior precision) and warm-start the next step from this solution.

\paragraph{MCMC.}
We run NUTS targeting the exact posterior over $\beta$ under the linear-Gaussian model, with 1000 warmup iterations and an initial 1000 retained draws, extending the run until the minimum effective sample size across coordinates reaches at least 2000. Standard diagnostics (ESS/$\hat R$) are used to monitor mixing, and the resulting draws are treated as the reference posterior for comparison.

\paragraph{Evaluation (coverage and efficiency).}
For each replicate, we compute marginal $90\%$ credible intervals for $\beta$ from each method’s draws (and, for the conjugate baseline, from the closed-form posterior $\mathcal{N}(m_n,S_n)$) and record coverage of the ground-truth $\beta^\star$ averaged over coordinates. We report mean coverage across replicates with a 95\% confidence interval based on the across-replicate standard error. For efficiency, we report minimum ESS/s for MCMC (min over coordinates) and paths/s for VPR, aggregated over replicates with 95\% confidence intervals.

\subsubsection{Logistic regression: highly-correlated feature simulation}\label{app:inference_logreg_sim}

We fit Bayesian logistic regression on the simulated design--response pair $(X,y)$ described above, with likelihood
\[
y_i \mid \beta \sim \mathrm{Bernoulli}\!\big(\sigma(x_i^\top \beta)\big),\qquad i=1,\dots,n,
\]
where $\sigma(t)=1/(1+e^{-t})$, and a spherical Gaussian prior
\[
\beta \sim \mathcal{N}(0,\tau_\beta^2 I_d),\qquad \tau_\beta=10.
\]

\paragraph{MF-VI.} We use a mean-field Gaussian variational family $q(\beta)=\prod_{j=1}^d \mathcal{N}(\beta_j\mid m_j,s_j^2)$ and maximise the ELBO by gradient ascent. To avoid Monte Carlo noise in the expected log-likelihood, we approximate the expected log-likelihood by Gauss--Hermite quadrature with $Q=20$ nodes; we also experimented with the deterministic softplus expectation bound of \citet{komodromos_logistic_2024} (Theorem~2.1 therein, truncated at level $\ell=12$), which is generally effective but noticeably more expensive in our setting. We optimise for $5{,}000$ Adam steps with learning rate $5\times 10^{-2}$ using full-data gradients, and use the resulting solution to initialise VPR.

\paragraph{VPR.} We run VPR with horizon $N=5{,}000$ and $1{,}000$ predictive-resampling paths, using the same diagonal Gaussian variational parameterisation as above. At each predictive-resampling step we (i) propagate covariates forward via the bootstrap resampling scheme applied to the observed rows of $X$, using a single shared resampled covariate stream across all paths, and (ii) impute the next label by sampling from the current variational predictive distribution induced by $q(\beta)$ and the logistic likelihood. We then update the variational parameters with $10$ warm-started Adam steps (learning rate $0.5$) using minibatches of size 15 drawn from the data available at that step: the newly imputed observation is always included, and the remaining batch elements are sampled uniformly from the previously available observations.
The variational expected log-likelihood term is rescaled by $(n+i+1)/b$ at each inner step, where $n+i+1$ is the size of the augmented dataset and $b$ is the minibatch size, so that the stochastic gradient remains an unbiased estimator of the full-dataset ELBO gradient.

\paragraph{MCMC.} We run NUTS targeting the exact posterior over $\beta$ under the logistic model, with $1{,}000$ warmup iterations and an initial $1{,}000$ retained draws, extending the run until the minimum effective sample size across coordinates reaches at least $1{,}000$. Standard diagnostics (ESS/$\hat R$) are used to monitor mixing, and the resulting draws are treated as the reference posterior for comparison.

\subsubsection{Logistic regression: real datasets}\label{app:inference_logreg_real}

We fit the same Bayesian logistic regression model as in \cref{app:inference_logreg_sim} on each of the five real datasets summarised
in \cref{tab:logreg_datasets}, retaining the spherical Gaussian prior $\beta \sim \mathcal{N}(0,\tau_\beta^2 I_d)$ with $\tau_\beta=10$. The inference setup mirrors \cref{app:inference_logreg_sim}; we describe only the differences.

\paragraph{MF-VI.} Same mean-field Gaussian variational family $q(\beta)=\prod_{j=1}^d \mathcal{N}(\beta_j\mid m_j,s_j^2)$. The expected log-likelihood is approximated by Gauss--Hermite quadrature with $Q=20$ nodes per coordinate. We optimise for $2{,}000$ Adam steps with learning rate $5\times 10^{-2}$ on the full training set, and use the resulting solution to initialise VPR.

\paragraph{VPR.} Same diagonal Gaussian variational parameterisation as
in MF-VI, initialised at the MF-VI solution. We use $L=1{,}000$ predictive-resampling paths and horizon $N=1{,}000$. At each predictive-resampling step we propagate covariates by the same bootstrap scheme as in \cref{app:inference_logreg_sim} (single shared resampled covariate stream across paths), impute the next label from the current variational predictive distribution, and update the variational parameters with $S$ warm-started Adam steps at learning rate $5\times 10^{-2}$ on a minibatch of size $100$ that always contains the newly imputed observation and $99$ uniformly sampled previously-available observations. We report results for $S=10$ in \cref{sec:logreg_datasets} in the main text and sweep $S\in\{2,10\}$ in \cref{app:logreg_real_ablations}.

\paragraph{MCMC.} Identical to \cref{app:inference_logreg_sim}: NUTS with $1{,}000$ warmup iterations and $1{,}000$ initial retained draws, extended until the minimum coordinate-wise ESS reaches at least $1{,}000$.

\paragraph{Evaluation (predictive accuracy, posterior similarity, and efficiency).}                     
For each train/test split we compare the methods on three axes. (i) \emph{Predictive accuracy}: we report the negative log predictive density (NLPD) on the held-out test set, expressed as the ratio $\mathrm{NLPD}_{\text{method}}/\mathrm{NLPD}_{\text{MCMC}}$ so values close to one indicate parity with the gold-standard MCMC predictive (lower is better). (ii) \emph{Posterior similarity}: since no ground-truth posterior is available, we treat the long-run NUTS draws as the reference and report the unbiased squared maximum mean discrepancy $\mathrm{MMD}^2$ \citep{gretton2012kernel} between each
method's posterior draws on $\beta$ and the MCMC reference, using an RBF kernel with bandwidth set by the median heuristic on the pooled, MCMC-whitened samples (lower is better). (iii) \emph{Wall-clock efficiency}: we report minimum ESS/s for MCMC (minimum across coordinates of $\beta$, divided by sampling time) and paths/s for VPR (number of completed predictive-resampling paths per second). All quantities are aggregated across the train/test splits and reported as the mean with a $95\%$ confidence interval based on the across-split standard error.

\subsubsection{Linear mixed random effects model}
We fit the LMRE model as described above. We place the priors
\(\beta \sim \mathcal{N}(0,\tau_\beta^2 I_d)\), \(\tau_\beta = 1\), and \(D \sim \operatorname{InvWishart}(\nu_0,S_0)\), \(\nu_0 = r + 2\), \(S_0 = I_r\), with group-specific random effects \(u_g \mid D \stackrel{\mathrm{iid}}{\sim} \mathcal{N}(0,D)\) for \(g = 1,\dots,G\).

\paragraph{MF-VI.}
We use a mean-field family that factorises across parameter blocks,
\(q(\beta,u,D) = q(\beta)\big(\prod_{g=1}^G q(u_g)\big)q(D)\), where \(q(\beta)\) and each \(q(u_g)\) are diagonal-covariance Gaussians, and \(q(D)\) is inverse-Wishart. We maximise the (full-data) ELBO with Adam for 5,000 steps using learning rate \(10^{-2}\). We use the resulting solution to initialise VPR.

\paragraph{VPR.}
We run VPR with horizon \(N=1000\) and 2000 predictive-resampling paths, using the same diagonal-Gaussian parameterisation for \(\beta\) and \(\{u_g\}\) and an inverse-Wishart factor for \(D\). At each predictive-resampling step we (i) propagate covariates forward by bootstrap-resampling from the observed rows of \((X,Z,\text{group index})\), using a single shared resampled covariate stream across all paths, and (ii) impute the next response by sampling from the current variational predictive distribution under the Gaussian observation model. We then update the variational parameters with $10$ warm-started Adam steps (learning rate \(10^{-2}\)) on a minibatch of size 50 drawn from the data available at that step; the newly imputed observation is not forcibly included in the optimisation minibatch. Since minibatching and unequal group sizes require care in hierarchical models, we optimise an unbiased minibatch ELBO estimator that appropriately reweights likelihood contributions and includes the correct KLD terms induced by the hierarchical prior \(u_g\mid D\) and the inverse-Wishart factor on \(D\).

\paragraph{MCMC.} We run NUTS on the full hierarchical model (including inference over \(D\)) with \(1{,}000\) warmup iterations and \(2{,}000\) retained samples. Mixing is monitored with standard diagnostics (ESS/\(\hat R\)), and the resulting draws are treated as the reference posterior when summarising coverage and efficiency.

\paragraph{Evaluation (coverage and efficiency).}
For each replicate, we compute marginal 90\% credible intervals for \(\beta\) and \(u\) from each method’s draws and record coverage of the ground-truth simulated values (averaged over coordinates, and over groups for \(u\)). We report the mean coverage over 100 replicates with a 95\% confidence interval based on the across-replicate standard error. For efficiency, we report minimum ESS/s for MCMC (min over \(\beta\) coordinates) and paths/s for VPR, again aggregated over replicates with 95\% confidence intervals.

\section{Additional results}
\subsection{Linear regression with optimal MF-VI updates} 
\paragraph{Most-correlated-plane (2D projection) diagnostic.}

To visualise posterior dependence in a high-dimensional parameter, we project posterior draws onto a two-dimensional subspace chosen to emphasise the strongest dependence implied by the exact posterior covariance. Concretely, using the exact posterior mean $m_n$ and covariance $S_n$, we select two orthonormal directions $u,v\in\mathbb{R}^d$ derived from $S_n$ and map each draw $\beta$ to the coordinates
\[
(w_1,w_2) \;=\; \big(u^\top \tilde{\beta},\, v^\top \tilde{\beta}\big),
\qquad
\tilde{\beta}_k \;=\; \frac{\beta_k-(m_n)_k}{\sqrt{(S_n)_{kk}}}.
\]
This yields a 2D view that concentrates on the most informative dependence directions under the true posterior geometry. The diagonal panels show the marginal densities of $w_1$ and $w_2$; the off-diagonal panel shows bivariate density contours. In \Cref{fig:most_correlated_plane} (shown for $d=10$, $n=30$), VPR and MCMC are nearly indistinguishable in both marginals and joint geometry in this projection, whereas MF-VI shows the expected mean-field distortion.

\begin{figure}
    \centering
    \includegraphics[width=0.5\linewidth]{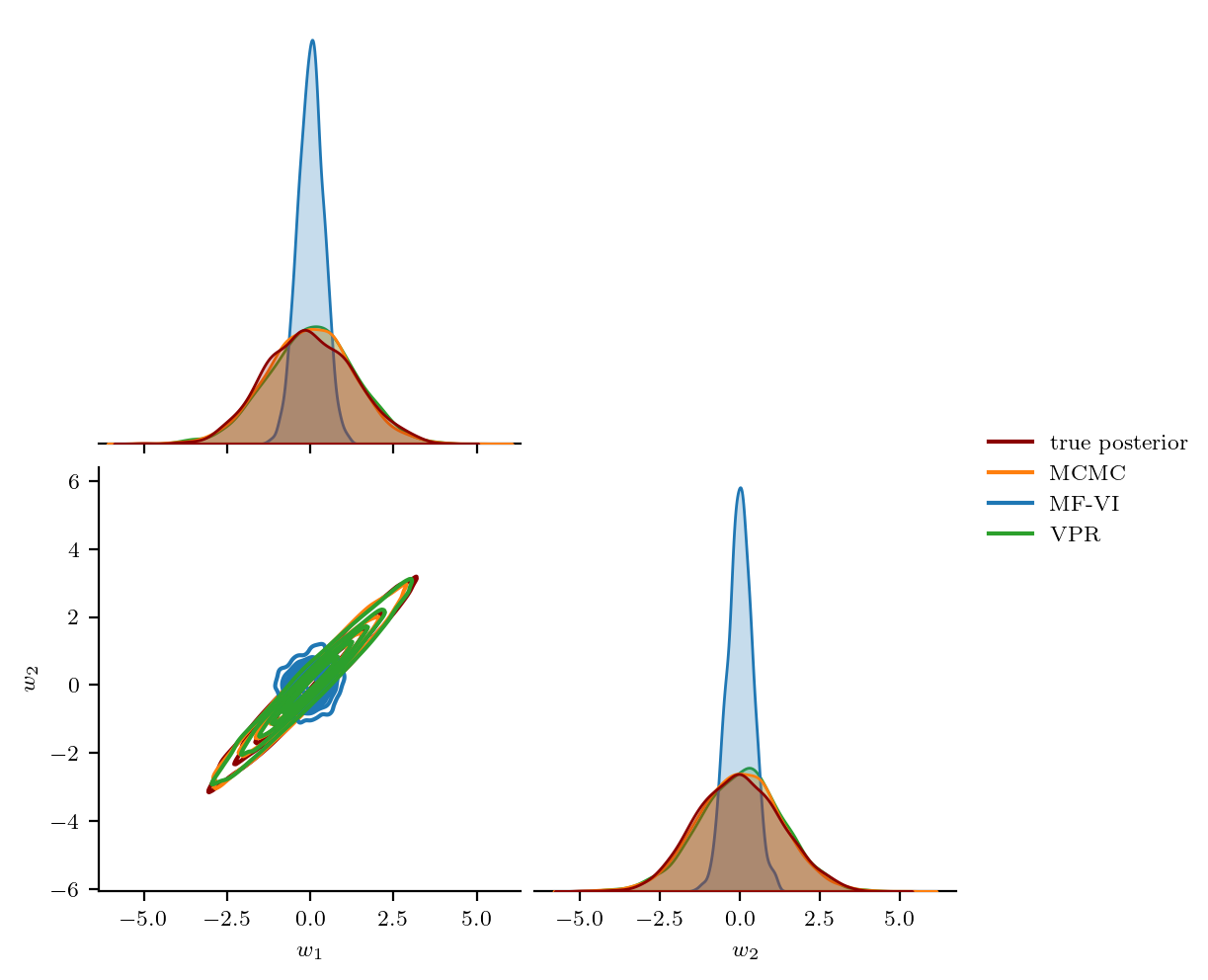}
    \caption{Most-correlated-plane projection ($d=10$, $n=30$). Posterior draws projected onto a 2D subspace selected from the exact posterior covariance.}
    \label{fig:most_correlated_plane}
\end{figure}

\subsection{Logistic regression}\label{app:logistic_experiments}

\subsubsection{Highly-correlated feature simulation}
\paragraph{Pathwise stability diagnostic.}
For the highly-correlated feature simulation generated as described in \cref{app:logistic_data} and discussed in \cref{sec:logreg_datasets},
\Cref{fig:logistic_jellyfish} visualises the VPR procedure by showing the evolution of VPR along the forward-resampling recursion.
Each line in the figure shows one of $1000$ resampling paths, reporting the value of one dimension of the variational mean parameter as a function of the resampling step.
While the trajectories exhibit high variability for the first few resampling steps, they quickly stabilise, indicating convergence.
Such a convergence is consistent with the intended sequential behaviour of VPR, supporting the use of a finite horizon $N$ as a practical approximation.

\begin{figure}[t]
    \centering
    \includegraphics[width=0.85\linewidth]{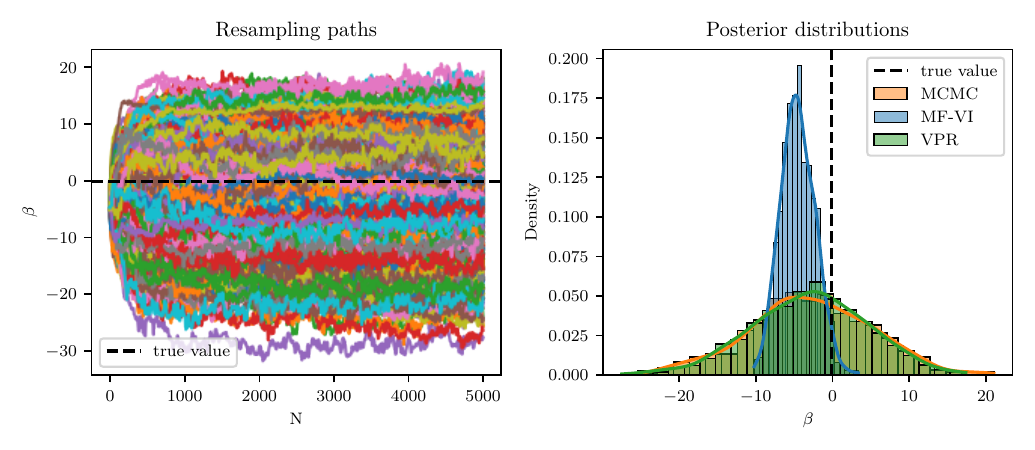}
    \caption{Evolution of the variational mean across predictive-resampling steps for a representative coefficient in Bayesian logistic regression.}
    \label{fig:logistic_jellyfish}
\end{figure}

\subsubsection{Real datasets} \label{app:logreg_real_ablations}

\paragraph{Ablation of $S$ for all datasets.}

\Cref{tab:logistic_real_data_ablated} reports VPR with inner-update $S=2$ and $S=10$ on all five datasets. Even at $S=2$, VPR's predictive performance is already on par with MCMC: NLPD ratios stay within roughly $1\%$ across all datasets, while MF-VI lags
noticeably on \texttt{skin}, and the $\mathrm{MMD}^2$ to MCMC is smaller than MF-VI's on every dataset, with the gap most pronounced on \texttt{skin} ($0.011$ vs $0.111$, roughly a $10\times$ reduction). Increasing $S$ to $10$ yields modest further reductions in $\mathrm{MMD}^2$ at a $\sim$$10\%$ cost in paths/s, with NLPD ratios essentially unchanged. VPR is therefore largely insensitive to $S$ over this range for the tested settings, and remains effective with very few inner updates per resampling step; $S=2$ already delivers most of the benefit at correspondingly higher throughput.
  
\begin{table*}[!t]
\centering
\caption{Logistic: Ablation of the VPR inner-update budget $S\in\{2,10\}$ across the five real datasets ($n_{\mathrm{tr}}=100$). Metrics match \cref{tab:logistic_real_data}: NLPD ratios to MCMC and  $\mathrm{MMD}^2$ between each method's posterior draws and the MCMC reference (lower is better), and wall-clock efficiency as minimum coordinate-wise ESS/s for MCMC and paths/s for VPR (higher is better). Entries are mean $\pm$ 95\% CI over $100$ random train/test splits; bold marks the best entry per metric within each dataset. MF-VI and MCMC do not depend on $S$ and are reported only on the first ($S=2$) row of each block.}
\label{tab:logistic_real_data_ablated}
\begin{tabular}{llcccccc}
\toprule
\multicolumn{2}{l}{} & \multicolumn{2}{c}{NLPD ratio} & \multicolumn{2}{c}{MMD$^2$ vs MCMC} & \multicolumn{2}{c}{Speed} \\
\cmidrule(lr){3-4}
\cmidrule(lr){5-6}
\cmidrule(lr){7-8}
Dataset ($d$) & $S$ & MF/MCMC & VPR/MCMC & MF-VI & VPR & ESS/s & paths/s \\ 
\midrule
german (20) & 2 & 1.008{\scriptsize$\pm$0.001} & 1.008{\scriptsize$\pm$0.002} & 0.012{\scriptsize$\pm$0.000} & 0.007{\scriptsize$\pm$0.000} & 202{\scriptsize$\pm$5} & \textbf{253{\scriptsize$\pm$
2}} \\
 & 10 &  & 1.008{\scriptsize$\pm$0.002} &  & \textbf{0.006{\scriptsize$\pm$0.000}} &  & 226{\scriptsize$\pm$1} \\
mozilla4 (5) & 2 & 1.003{\scriptsize$\pm$0.001} & 1.001{\scriptsize$\pm$0.001} & 0.019{\scriptsize$\pm$0.001} & 0.004{\scriptsize$\pm$0.001} & 186{\scriptsize$\pm$6} & \textbf{274{\scriptsize$\pm
$2}} \\
 & 10 &  & 1.001{\scriptsize$\pm$0.001} &  & \textbf{0.003{\scriptsize$\pm$0.000}} &  & 247{\scriptsize$\pm$2} \\
phoneme (5) & 2 & 1.001{\scriptsize$\pm$0.000} & 1.001{\scriptsize$\pm$0.001} & 0.007{\scriptsize$\pm$0.001} & 0.005{\scriptsize$\pm$0.001} & 249{\scriptsize$\pm$8} & \textbf{275{\scriptsize$\pm$
2}} \\
 & 10 &  & 1.002{\scriptsize$\pm$0.001} &  & \textbf{0.004{\scriptsize$\pm$0.001}} &  & 246{\scriptsize$\pm$1} \\
skin (3) & 2 & 1.033{\scriptsize$\pm$0.003} & 1.010{\scriptsize$\pm$0.002} & 0.111{\scriptsize$\pm$0.006} & 0.011{\scriptsize$\pm$0.001} & 143{\scriptsize$\pm$3} & \textbf{278{\scriptsize$\pm$2}}
 \\
 & 10 &  & \textbf{1.006{\scriptsize$\pm$0.001}} &  & \textbf{0.006{\scriptsize$\pm$0.001}} &  & 250{\scriptsize$\pm$2} \\
telescope (10) & 2 & \textbf{0.998{\scriptsize$\pm$0.001}} & 1.001{\scriptsize$\pm$0.001} & 0.045{\scriptsize$\pm$0.001} & 0.007{\scriptsize$\pm$0.000} & 173{\scriptsize$\pm$5} & \textbf{273{\scriptsize$\pm$2}} \\
 & 10 &  & 1.002{\scriptsize$\pm$0.001} &  & \textbf{0.005{\scriptsize$\pm$0.000}} &  & 244{\scriptsize$\pm$2} \\
\bottomrule
\end{tabular}
\end{table*}

\paragraph{Ablation of $N,S$ for \texttt{german}.}
\Cref{fig:german_ablation_frontier} plots posterior quality ($\text{MMD}^2$ to MCMC) against wall-clock time for $S \in \{1,2,10\}$ and different $N$ on \texttt{german}. Performance improves rapidly with $N$ at first, then plateaus around $N=800-1000$, with similar posterior quality for all tested $S$. Hence the practical gains of VPR are not tied to a single finely tuned $(N,S)$: several moderate settings achieve near-MCMC quality well before the MCMC wall-clock time (red dashed line). Once $N$ is moderately large, increasing $S$ beyond $S=1-2$ or further increasing $N$ yields only marginal gains at higher cost.

\paragraph{Ablation of $L$, $b$ for \texttt{german}.} 
\Cref{tab:german_ablation} varies batch size in the analysis of \texttt{german} and shows that very small $b$ (e.g. $10$) can hurt performance, but moderate values (e.g. $50$) give a good speed/accuracy trade-off, consistent with standard stochastic VI behaviour. It also ablates the number of paths $L$, and shows that as $L$ increases, posterior quality remains unchanged while efficiency increases, demonstrating that VPR benefits from parallelisation across paths and suggesting computational advantages over MCMC beyond those shown in the main results on our hardware.

\begin{figure}[t]                            
  \centering
  \resizebox{\linewidth}{!}{
  \begin{subfigure}[b]{0.49\linewidth}
      \centering 
      \vspace{0pt}
      \includegraphics[width=\linewidth]{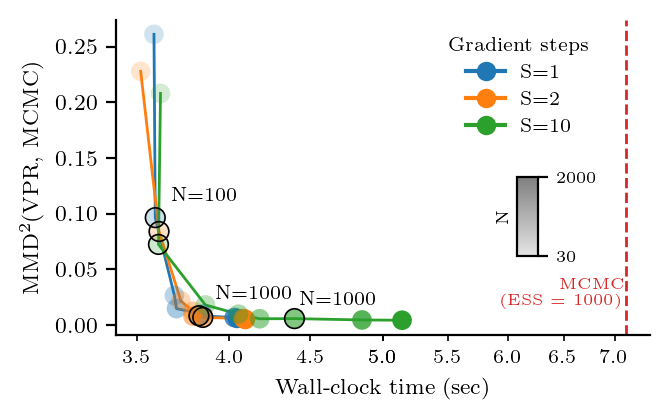}
      \caption{Accuracy--cost frontier}
      \label{fig:german_ablation_frontier}     
  \end{subfigure}
  \hfill                                
  \begin{subfigure}[b]{0.49\linewidth}
      \centering
      \vspace{0pt}
      \footnotesize
      \begin{tabular}{llc
c}
\toprule
ablation & setting & $\text{MMD}^2$ & 
paths/s \\
& & (MCMC)
\\
\midrule
{\scriptsize baseline} & {\scriptsize S=10, N=1K,} & 0.006 
         & 228 \\
         & {\scriptsize L=1K, b=100} & {\scriptsize± 0.001} 
         & {\scriptsize± 3}
         \\
\midrule
L & 5000 & 0.006 {\scriptsize± 0.001} & 
767 {\scriptsize± 7} \\
 & 10000 & 0.006 {\scriptsize± 0.001} & 
 1102 {\scriptsize± 9} \\
b & 10 & 0.040 {\scriptsize± 0.003} & 
257 {\scriptsize± 5} \\
 & 50 & 0.010 {\scriptsize± 0.001} & 
 241 {\scriptsize± 4} \\
\bottomrule
\end{tabular}
       \vspace{2.5em}
      \caption{Hyperparameter ablation}         
      \label{tab:german_ablation}
  \end{subfigure}   
  }
  \caption{German Statlog ($n{=}100$, $d{=}20$). (a): Accuracy--cost frontier for VPR  ($L{=}1000$ paths). Each curve fixes the number of gradient steps $S$; marker opacity increases with the resampling horizon $N \in \{30, 100, 500, 800,1000,1500, 2000\}$. Black-bordered markers highlight $N{=}100$ and $N{=}1000$. The dashed red line marks the wall-clock time for MCMC to reach ESS${=}1000$. (b): Sensitivity to hyperparameters; each row varies one setting from the baseline (top). Mean $\pm$ 95\% CI over 20 train/test splits.}             
  \label{fig:german_ablation}
\end{figure}

\subsection{Linear mixed effects}\label{app:lmre_experiments}
\paragraph{KDE pair plots (fixed and group-specific effects).}
To assess posterior agreement in the mixed-effects model, we report KDE pair plots for the joint posterior over the fixed effects $\beta$ and the group-specific random effects $u_g$ (one plot per group; see \cref{fig:lmre_pairplot_g1,fig:lmre_pairplot_g2,fig:lmre_pairplot_g3,fig:lmre_pairplot_g4}). Within each plot, the lower triangle shows bivariate KDE contours for MF-VI, VPR, and MCMC, while the upper triangle summarises dependence via Pearson correlations comparing MCMC versus VPR. Across groups and parameter pairs, VPR closely matches the MCMC reference in both dispersion and dependence structure, whereas MF-VI shows the expected mean-field artefacts and fails to reproduce key dependencies.

\begin{figure}[t]
    \centering
    \includegraphics[width=0.95\linewidth]{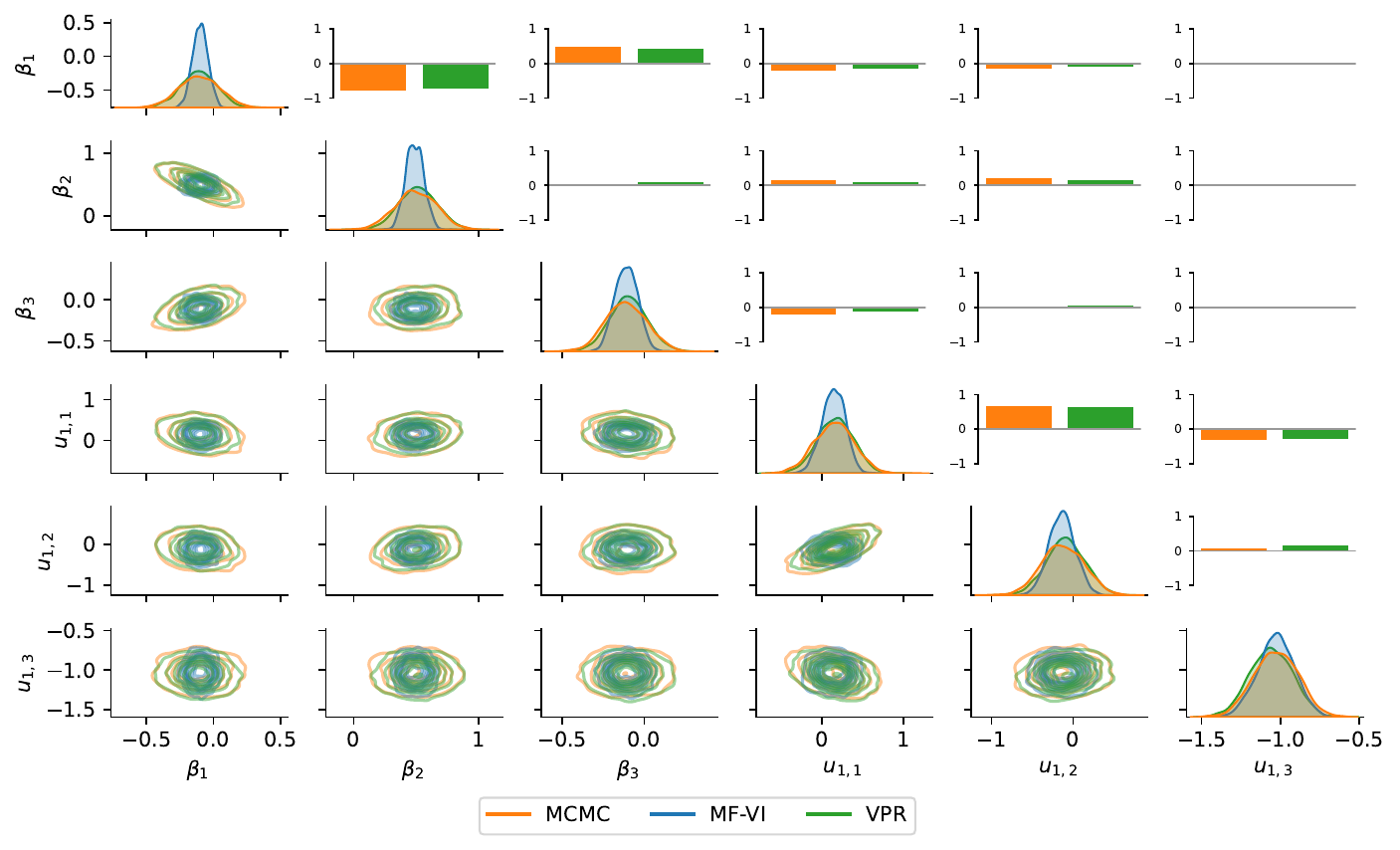}
    \caption{KDE pair plot for $(\beta, u_1)$. Lower: bivariate KDE contours for MF-VI/VPR/MCMC. Upper: Pearson correlations for MCMC vs VPR.}
    \label{fig:lmre_pairplot_g1}
\end{figure}

\begin{figure}[t]
    \centering
    \includegraphics[width=0.95\linewidth]{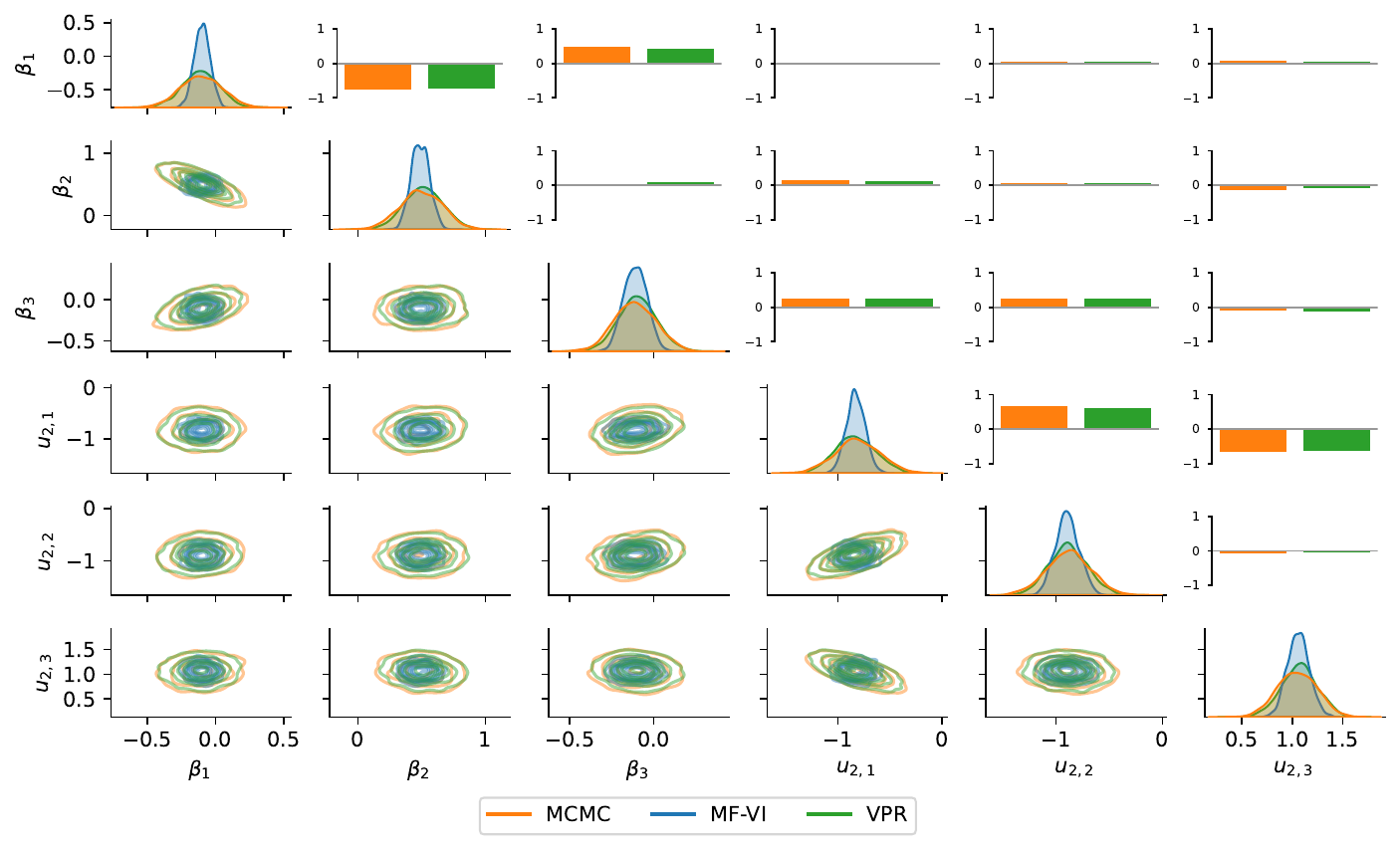}
    \caption{KDE pair plot for $(\beta, u_2)$. Lower: bivariate KDE contours for MF-VI/VPR/MCMC. Upper: Pearson correlations for MCMC vs VPR.}
    \label{fig:lmre_pairplot_g2}
\end{figure}

\begin{figure}[t]
    \centering
    \includegraphics[width=0.95\linewidth]{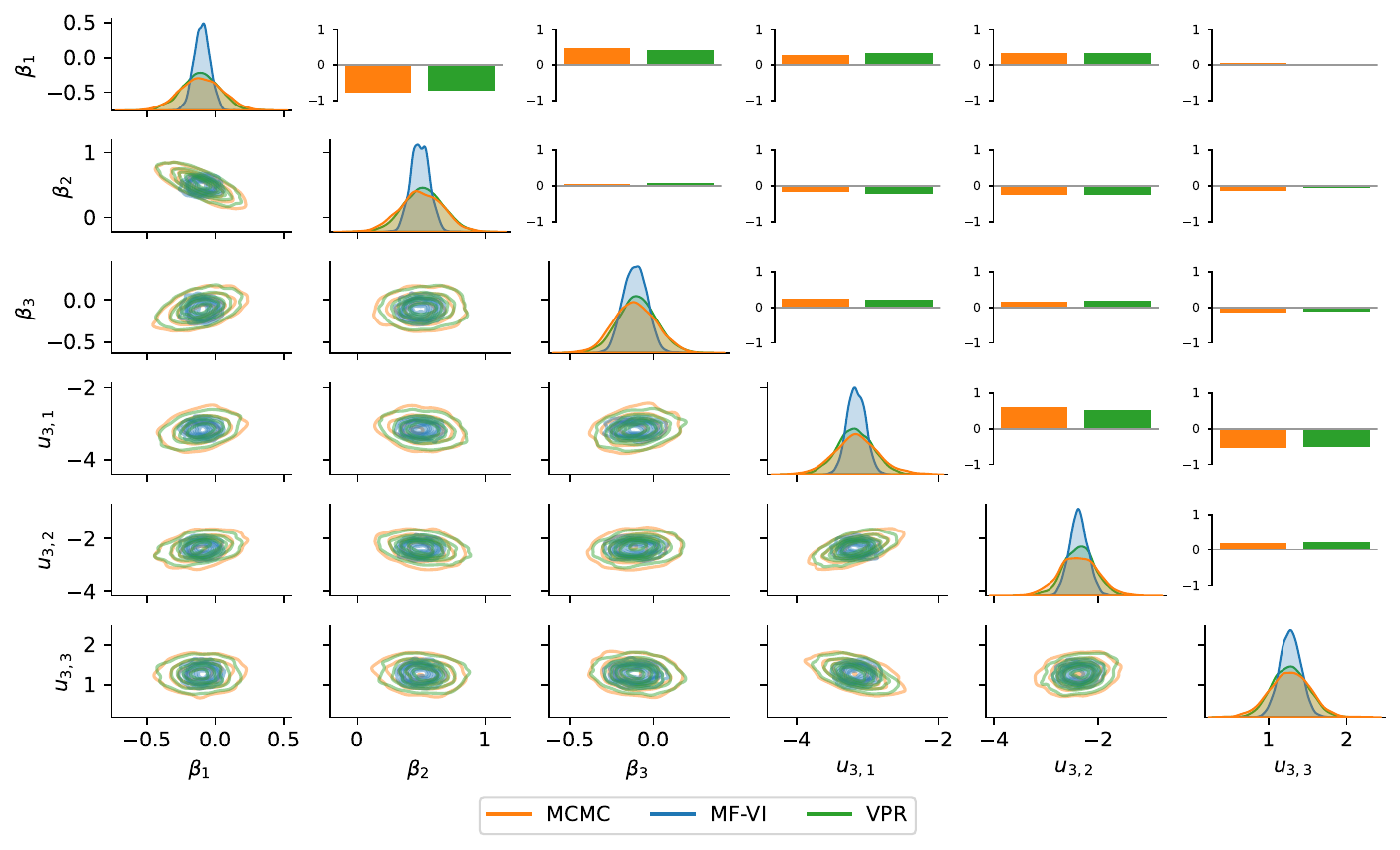}
    \caption{KDE pair plot for $(\beta, u_3)$. Lower: bivariate KDE contours for MF-VI/VPR/MCMC. Upper: Pearson correlations for MCMC vs VPR.}
    \label{fig:lmre_pairplot_g3}
\end{figure}

\begin{figure}[t]
    \centering
    \includegraphics[width=0.95\linewidth]{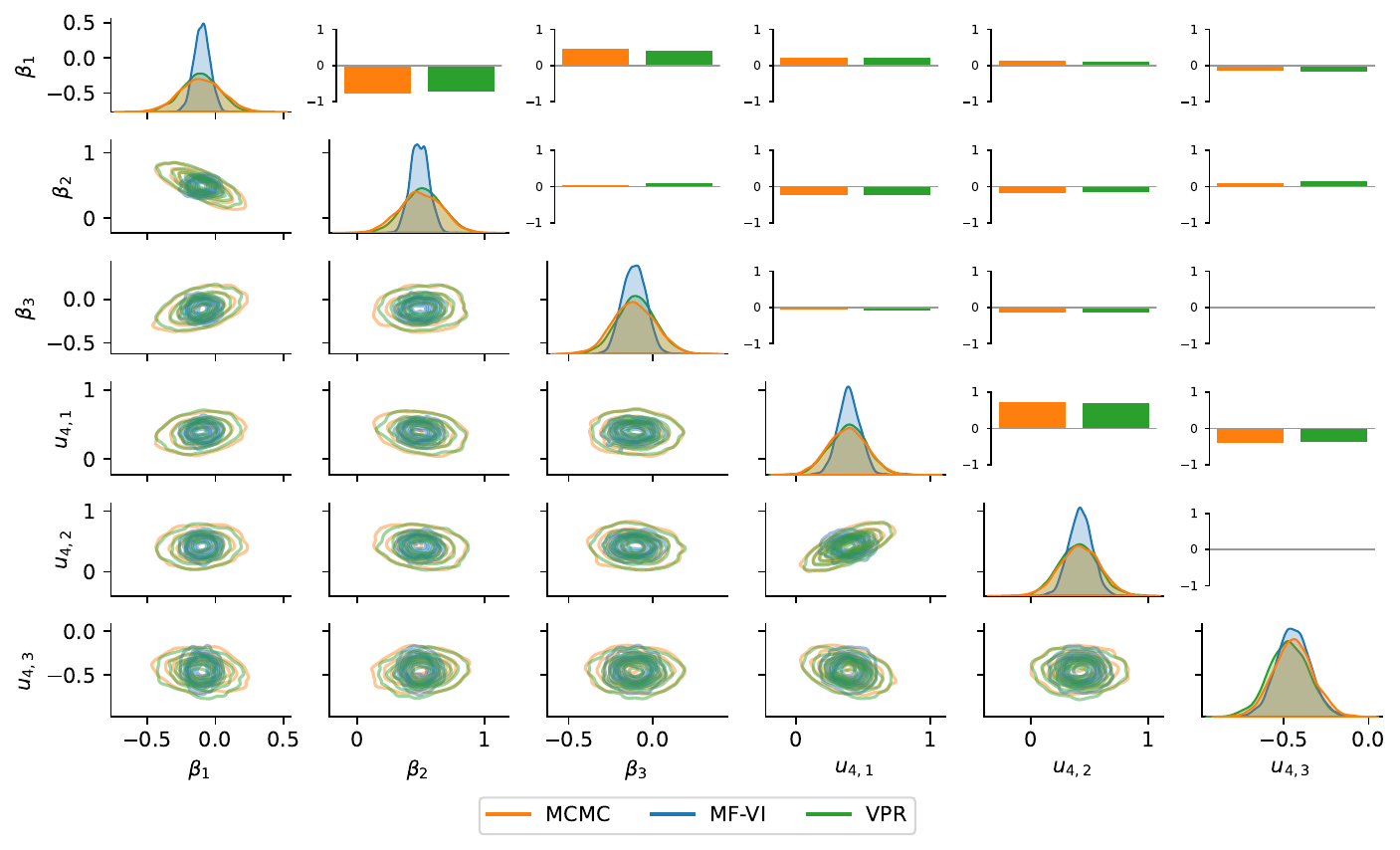}
    \caption{KDE pair plot for $(\beta, u_4)$. Lower: bivariate KDE contours for MF-VI/VPR/MCMC. Upper: Pearson correlations for MCMC vs VPR.}
    \label{fig:lmre_pairplot_g4}
\end{figure}

\end{document}